\documentclass[10pt,aps,prb,twocolumn,superscriptaddress]{revtex4-2}

\usepackage[utf8]{inputenc}
\usepackage[english]{babel}
\usepackage{mathtools}
\usepackage{physics}
\usepackage{braket}
\usepackage{bbm}
\usepackage{bm}
\usepackage{amsbsy}
\usepackage{amsthm}
\usepackage{amssymb}
\usepackage{amsfonts}
\usepackage{amsmath}
\usepackage{dsfont} 
\usepackage{graphicx, caption} 
\usepackage{subcaption}
\usepackage{epsfig}
\usepackage{epstopdf}
\usepackage{dsfont}
\usepackage{multibib}
\usepackage{xcolor}
\usepackage{color}
\usepackage{verbatim}
\usepackage{nomencl}
\usepackage[colorlinks]{hyperref}
\usepackage{float}
\usepackage[normalem]{ulem}
\usepackage{comment}

\newtheorem{prop}{Proposition}[section]

\begin{document}

    \title{Optimal observables for quantum-enhanced sensing and applications in a Floquet time crystal sensor}

    \author{M. A. Manya}\thanks{These authors contributed equally to this work.}
	\affiliation{Instituto de F\'isica, Universidade Federal Fluminense, Av. Gal. Milton Tavares de Souza s/n, Gragoat\'a, 24210-346 Niter\'oi, Rio de Janeiro, Brazil}

    \author{Andrei Tsypilnikov}\thanks{These authors contributed equally to this work.}
	\affiliation{Instituto de F\'isica, Universidade Federal Fluminense, Av. Gal. Milton Tavares de Souza s/n, Gragoat\'a, 24210-346 Niter\'oi, Rio de Janeiro, Brazil}

    \author{Fernando Iemini }
	\affiliation{Instituto de F\'isica, Universidade Federal Fluminense, Av. Gal. Milton Tavares de Souza s/n, Gragoat\'a, 24210-346 Niter\'oi, Rio de Janeiro, Brazil}

    \begin{abstract}
        In this work, we discuss how to determine and implement  feasible optimal observables for a metrology protocol that saturates the quantum Fisher information (QFI) bound. In particular, we focus our study on a simple protocol, namely the method of moments (MoM). We first demonstrate that the symmetric logarithmic derivative (SLD) operator, a Hermitian observable, once implemented in the MoM, saturates the QFI bound. However, the SLD is generally too complex and typically non-local, rendering its direct experimental realization unfeasible. To overcome this limitation, we explore its structure in a specific sensing model - a Floquet time crystal (FTC) acting as an ac field sensor - and show that the SLD can be approximated by substantially simpler observables, such as the bare spin magnetization or a parity observable, for different relevant initial state preparations. We further corroborate our theoretical predictions in a nuclear magnetic resonance system operating as an FTC sensor, employing experimentally motivated parameters to simulate its performance in a state-of-the-art implementation. In general, our results establish a practical route toward near-optimal metrology in FTC sensors, where the inaccessible SLD operator can be replaced by simpler observables while retaining quantum-enhanced sensitivity.
    \end{abstract}

	\pacs{}
	\maketitle

\section{Introduction}
\label{sec.introduction}

Quantum metrology plays a crucial role in the rapid development of quantum technologies and modern science~\cite{Toth_2014}. It harnesses the laws of quantum mechanics to achieve unprecedented precision and sensitivity in unknown parameter estimation~\cite{MONTENEGRO20251,Pirandola2018}. Fundamentally, it drives the creation of quantum sensors~\cite{RevModPhys.89.035002} capable of surpassing the limits of classical physics.

The primary goal of quantum metrology  is the study of accurate estimation of an unknown parameter $h$ encoded in the quantum state of a physical system. The parameter estimation, fundamentally, is limited by the quantum Cramér–Rao bound (QCRB), which establishes a lower bound on the variance of any unbiased estimator, through the quantum Fisher information (QFI)~\cite{npj_Yu2022}. However, a significant gap often exists between this theoretical limit and experimentally achievable precision. Reaching the QCRB requires measuring, in the metrology protocol, an optimal observable whose decoded parameter $h$ has an estimated uncertainty saturating the QFI bound. The determination of such an optimal observable, however, is not usually trivial.
Moreover, even if determined, implementing measurements corresponding to such an optimal observable may still be challenging, as it may lack a direct or experimentally accessible physical realization. Consequently, a central objective in practical quantum metrology is also to search for ways to approximate the theoretical precision limits within experimentally feasible measurement strategies.

Beyond equilibrium systems, non-equilibrium quantum platforms have recently emerged as promising candidates for enhanced quantum sensing and metrology \cite{MONTENEGRO20251}. In particular, Floquet time crystals (FTC) have proven to be a promising avenue for enhanced metrology, with different works showing their improved performance both in closed and dissipative systems~\cite{Andrei2026, PhysRevA.109.L050203, Andr10.21468/SciPostPhys.18.3.100, Andr7m63-lnb8, AndrArumugam2025, Andrbiswas2025floquetcentralspinmodel, Andrcabot2026parameterestimationonetwotime, AndrMontenegro2023, Andrmoon2024discretetimecrystalsensing, AndrPavlov_2023, AndrPhysRevB.111.024315, AndrPhysRevB.111.125159, AndrPhysRevLett.132.050801, AndrRondin_2014, Andrzybb-vxfz,yousefjaniNonlinearityenhanced2026}.
These novel phases of matter were first predicted theoretically~\cite{colloquium, PhysRevLett.118.030401} and subsequently observed experimentally~\cite{Zhang2017, Choi2017}. These systems can be studied within the framework of many-body physics, where quantum entanglement can be controlled and exploited to model atomic ensembles as quantum sensors, including nuclear magnetic resonance (NMR) platforms~\cite{pal2018temporal, Uhlig2019} and nitrogen-vacancy (NV) centers in diamond~\cite{PhysRevB.97.184301}. Importantly, the relevance of FTC to quantum metrology lies in their intrinsic stability and robustness against decoherence, arising from their double-period subharmonic response by the disorder and external perturbations. This noise resilience enables long-lived coherent dynamics, allowing measurements to be performed over extended interrogation times with improved precision.

This research addresses the central challenge of bridging the gap between the theoretical precision limits imposed by the QCRB and their practical realization in metrology protocols. Specifically, we ask: What is the optimal observable for estimating an unknown parameter $h$, both from a theoretical standpoint as well as in terms of experimental feasibility? To tackle this question, we consider a simple protocol, the method of moments (MoM), and show that the symmetric logarithmic derivative (SLD) operator~\cite{helstrom1969quantum} saturates the theoretical QFI bound. However, the SLD is generally complex and non-local, making its direct implementation challenging. To overcome this limitation, we explore its structure in FTC-based ac field sensors, and demonstrate how the SLD simplifies to experimentally feasible observables, such as bare spin magnetization or parity, for different initial state preparations.

We further explore our theoretical predictions in a state-of-the-art  FTC-NMR sensor. NMR platforms are powerful, well-established systems that have become increasingly relevant for quantum metrology, particularly for the study of complex spin dynamics and high-precision parameter estimation in ensemble-based systems~\cite{JONES201191}. In such quantum sensing schemes, sensitivity is governed by the collective behavior of nuclear spins, with NMR measurements providing access to unknown magnetic field parameters encoded in the spin precession frequency~\cite{Giovannetti2011}. To capture the essential features of sensor dynamics and spin interactions, we employ the star topology model~\cite{Mahesh_2021, PandePhysRevA.96.012330}, in which a central spin interacts symmetrically and collectively with an ensemble of auxiliary spins.

This manuscript is organized as follows: In the Sec.~\ref{sec:metrology}, we present the theoretical background of metrology, specifically QFI and MoM. In Sec.~\ref{sec:sld-operator-(andrei)} we present a theoretical development of the (SLD) optimal observable for the MoM. In the Sec.~\ref{sec:SLD in Floquet Time Crystal} an analytical expansion of SLD operator in a FTC phase is presented. In Sec.~\ref{sec:NMR platform} we present our NMR sensor model and a numerical analysis of its QFI and MoM. Lastly, we conclude in Sec.~\ref{sec.conclusions}.

\section{Metrology: QFI and Method of Moments}
\label{sec:metrology}

From a general metrological perspective, applicable to both classical and quantum regimes, any parameter estimation protocol consists of four essential steps:
(i) preparation of a quantum probe, (ii) dynamical encoding of the unknown parameter, (iii) information extraction via measurement, and (iv) classical inference to process the unknown parameter.

Classically, the parameter estimation is fundamentally limited by statistical fluctuations arising from independent measurement events. When the probe consists of uncorrelated resources such as photons~\cite{Muiretal2012}, atoms, or spins, the measurement statistics follow a Poissonian distribution, leading to the so-called shot-noise limit (SNL). Under these conditions, the uncertainty in the estimation of an unknown parameter scales with the number of resources, in general described by the number of spins $N$ and the interrogation time $t$, as $ 1/\sqrt{Nt}$. The SNL represents the ultimate precision bound achievable within classical physics and has been extensively verified in optical~\cite{Hobbs:90} and atomic~\cite{Rocco_2014} interferometric experiments~\cite{helstrom1969quantum}. In the quantum frame, metrology exploits quantum coherence and entanglement to achieve measurement precision beyond the classical limit~\cite{PhysRevLett.59.278, Toth_2014}. In particular, quantum mechanics enables optimal precision with uncorrelated quantum probes, called the standard quantum limit (SQL)~\cite{science_vitorio2004}, which scales as $1/\sqrt{Nt^2}$, while correlated quantum probes surpass the SQL and approach the ultimate Heisenberg limit~\cite{prl_liu2025}, characterized by a $1/Nt$ scaling.

\subsection{Quantum Fisher Information}

The QFI is a fundamental physical quantity that quantifies the sensitivity of a quantum state to changes in some physical parameter in the dynamics, such as the magnetic field or temperature.  In metrology,
the QFI quantifies the maximum achievable sensitivity of
a quantum sensor, whose  precision is bounded by the QCRB limit:
\begin{equation}
\label{eq:01}
\mathrm{var}(h) \ge \frac{1}{F_h}.
\end{equation}
Here, $F_h$ is the QFI and  $\mathrm{var}(h) = \langle(h_{\rm est} - h )^2\rangle $ the estimation uncertainty,
where $h_{\rm est}$ is the estimated parameter using the metrology protocol within an ``estimator method'' (likelihood, method of moments, among others) and $h$ is its exact value.
 There are equivalent forms to represent the QFI which will be relevant for our analysis, as we discuss. First, the  QFI can be expressed through the SLD operator $\hat{L}_h$, which is an operator satisfying the following relation,
 \begin{equation}
        \label{eq:03}
        \frac{\partial \hat{\rho}_h}{\partial h}
        = \frac{1}{2}\left( \hat{\rho}_h \hat{L}_h + \hat L_h \hat\rho_h \right),
    \end{equation}
with  $\hat{\rho}_h$ the quantum state with encoded $h$ parameter. The QFI in terms of the SLD~\cite{helstrom1969quantum} is given by, 
\begin{equation}
   \label{eq.QFI.SLD}
    F_{h}[\hat{\rho}_h] = \mathrm{Tr}\left[\hat{\rho}_h \hat{L}_h^{2}\right].
\end{equation}

Alternatively, for systems in a pure state, the QFI reduces to a more intuitive expression,
\begin{equation}
    \label{eq.QFI.purestates}
    F_h(t) = 4\left( \langle \partial_h\psi_h(t)|\partial_h\psi_h(t)\rangle  - |\langle \partial_h\psi_h(t)|\psi_h(t)\rangle|^2 \right),
\end{equation}
where  $\partial_h |\psi_h(t)\rangle = \frac{\partial}{\partial h} |\psi_h(t)\rangle $. By expanding this expression and working in the Heisenberg picture, the QFI can alternatively be rewritten in terms of the variance of an evolved operator, the so called Heisenberg signal operator (HSO), $\hat S_h(t)$, as follows:
\begin{equation}
    \label{eq.QFI.HSO.purestates}
    F_h(t) = 4 \mathrm{var}(\hat{S}_h(t)),
\end{equation}
with $\mathrm{var}(\hat O) = \langle \psi_h(0)|\hat O^{2}|\psi_h(0)\rangle -
    \langle \psi_h(0)|\hat O|\psi_h(0)\rangle^2$.
The definition of the HSO operator comes from the Duhamel formula~\cite{Dyson1949, PhysRevA.109.L050203}
    \begin{equation}
        \hat{S}_h(t) \equiv \int_0^t \hat{U}_h^\dagger(t')\left( \frac{\partial \hat{H}(t')}{\partial h} \right)\hat{U}_h(t') dt',
        \label{eq:HSO}
    \end{equation}
which satisfy the differential equation of the parameter derivative of the evolution operator,
   \begin{equation}
        \frac{\partial \hat U_h (t) }{\partial h} = -i \hat U_h (t) \hat{S}_h(t).
        \label{eq.qfi_def}
    \end{equation}

Within these three different expressions of the QFI, one can clearly see how it characterizes the system's response through quantum state changes, as shown in Eq.~\eqref{eq.QFI.purestates}, and how it can be understood as a statistical fluctuation inherent to the variance of the HSO Eq.~\eqref{eq.QFI.HSO.purestates}. In both instances, these changes are driven by variations in the external field $h$.

Although QFI establishes the maximum theoretical precision for estimating unknown parameters, this limit may be   unattainable in practice and different estimation strategies are employed in order to circumvent this issue. A simple strategy is based on the method of moments, which is based on the dynamics of the average expectation value of an (possibly optimal) observable, as we discuss below.


\subsection{Method of moments}
\label{sec.mom}
In a practical scenario, the connection between the quantum state and the experimental outcomes is formally established through positive operator-valued measure (POVM), defined by a set of operators $\{ \hat\Pi_m \}$ that satisfy $\sum_m\hat{\Pi}_m = \mathbb{I}$.  Within this framework, the average value $\bar{m}_k$ of $k$ measurement results $\{ m_1, m_2 ... m_k\}$ is linked to the expectation value $\mu_k(h)$ of the corresponding observable $\hat{O}$. Specifically, $\mu(h) = \text{Tr}[\hat \rho_h \hat{O}] = \sum_m m\, P(m|h)$, where $P(m|h) = \text{Tr}[\hat \rho \hat \Pi_m]$ is the probability of obtaining result $m$ given the POVM.  The MoM assumes that for a large number of samples, the sample average $\bar {m} _ k = \frac{1}{k}\sum^k_i m_i$ converges to the theoretical expectation value $\mu(h)$. Consequently, the estimator $h$ is derived by solving $\mu_k(h) = \bar m_k$. By applying a first-order Taylor expansion~\cite{LucaPezzeRevModPhys.90.035005} on $\mu_k(h)$ around the true value, the estimation uncertainty can be expressed $\text{var}(h) = 1/\text{MoM}(h, t)$ with  the estimator,
    \begin{equation}
		\label{eq:mom1}
		\text{MoM}(h, t) = \frac{ \left( \partial_h \langle \hat{O}
		\rangle_h \right)^2 }{\text{Var}(\hat {O})_h} = \frac{\left(
		\partial_h \langle \hat{O} \rangle_h \right)^2}{\langle
	\hat{O}^2 \rangle_h - \langle \hat{O} \rangle^2_h}.
	\end{equation}

The observable $\hat{O}$ is, in principle, arbitrary;
 however its choice is crucial, as it must be capable of maximally probing the fluctuations induced by the external field $h$. These fluctuations are captured by the term $|\partial_h \langle\hat{O}\rangle_h|^2$, whereas $\text{var}(\hat{O})_h$ represents the intrinsic quantum noise of the system. In other words, the MoM can be viewed as a signal-to-noise ratio where the signal identifies the regime of highest sensitivity, and the noise constrains the ultimate precision of parameter estimation.

 We further notice that, as the maximization  of the MoM requires optimizing an observable $\hat O$, it could be thought in terms of the SLD operator, which is directly related to the QFI through Eq.~\eqref{eq.QFI.SLD}. We explore this connection in the following section.
  It is important to note, however, that although the norm of the SLD generally grows with both time and system size, the observable associated with the MoM does not necessarily need to obey the same scaling behavior. This stems from the fact that the MoM is invariant under a global rescaling of the observable operator, $\hat O \to c\,\hat O$, for any $c \neq 0$. Consequently, the relevant quantities are not the absolute magnitude of the operator, but rather its unitary transformation and susceptibility across the Hilbert space.

    \section{(SLD) Optimal observable for method of moments}
    \label{sec:sld-operator-(andrei)}
    \label{andrei}

    In this section, we analyze, in theory, protocols that maximize the precision of the MoM. First, we show that the MoM can always achieve its maximum precision (saturating the QFI bound) by employing the SLD operator as its observable. Next, we discuss a general protocol and highlight restrictions from an experimental point of view that could prevent its direct use, as well as possible forms to overcome it. We illustrate here its use in a simple single spin Floquet dynamics, and leave the discussion of the more intricate case of a FTC phase to the next section.

    \subsection{Saturation of MoM}

    To demonstrate that the SLD operator defined by Eq.~\eqref{eq:03} saturates the QFI bound for pure states within the  MoM, it is helpful to write it in the following form~\cite{paris_quantum_2009, liu_quantum_2020},
    \begin{equation}
        \label{eq:SLD.pure.states}
        \hat{L}_h = 2 \frac{\partial \rho_h}{\partial h} = 2\Big(\ket{\partial_h \psi_h} \bra{\psi_h} + \ket{\psi_h} \bra{\partial_h \psi_h}\Big)
    \end{equation}
    One can now introduce a new operator $\hat{L}_{h_0}(t) \equiv \hat{L}_{h}(t)\big|_{h=h_0}$, which is the SLD operator computed for a fixed value $h_0$.  Two properties make this the natural observable for the MoM: unlike $\hat{L}_h$, it can be constructed without prior knowledge of the unknown $h$; being independent of $h$, its derivative $\partial_h\hat{L}_{h_0}=0$, so the sensitivity of the expectation value $\partial_h\langle\hat{L}_{h_0}\rangle_h$ is sourced entirely by the state $\hat\rho_h$. In that case Eq.~\eqref{eq:mom1} gives for $\hat O \equiv {\hat L}_{h_0}$
    \begin{equation}
        \label{eq:MomSLD}
        \text{MoM}(h,t) = \frac{
            \Tr^2 {\left(  {\hat L}_{h_0} \partial_h \rho_h \right)}}{\langle
        {\hat L}_{h_0}^2 \rangle - \langle {\hat L}_{h_0} \rangle^2}
    \end{equation}
 where we used that $\partial_h \langle {\hat L}_{h_0} \rangle = \partial_h \left(\mathrm{Tr}( \hat{\rho}_h {\hat L}_{h_0})\right) = \mathrm{Tr}({\hat L}_{h_0}\, \partial_h \hat{\rho}_h)$, since ${\hat L}_{h_0}$ is independent of $h$.

 From the definition of Eq.~\eqref{eq:03}, if the observable was chosen as the SLD evaluated at the true (unknown) value $h_0 \to h$ one would obtain  that
 	\begin{equation}
 		\begin{aligned}
 		\lim_{h_0 \rightarrow h} \Tr  \left(  {\hat L}_{h_0} \partial_h \rho_h \right)
 		&= \lim_{h_0 \rightarrow h}
         \Tr  \left(  {\hat L}_{h_0} \frac{1}{2}\left( \hat{\rho}_h \hat{L}_h + \hat L_h \hat\rho_h \right) \right) \\
 		&= \lim_{h_0 \rightarrow h} \Tr  \left(  \frac{1}{2}({\hat L}_{h_0} \hat L_h
         + {\hat L}_h \hat L_{h_0})\hat\rho_h \right) \\
 		&= \langle \hat L^2_{h} \rangle.
 		\end{aligned}
 	\end{equation}
 	Further noticing that
    \begin{eqnarray}
    \langle \hat{L}_{h} \rangle &=&    2\langle \psi_h | \Big(\ket{\partial_h \psi_h} \bra{\psi_h} + \ket{\psi_h} \bra{\partial_h \psi_h}\Big) | \psi_h\rangle   \nonumber \\
    &=&
    2 \partial_h \Big( \langle \psi_h|\psi_h\rangle \Big) \nonumber \\
    &=& 0,
    \end{eqnarray}
    since the state norm is preserved,
    we see that Eq.~\eqref{eq:MomSLD} saturates the QFI. Precisely, given
    \begin{equation}
        \label{eq:11}
        \hat O  \equiv \hat L_{h_0}: \quad \lim_{   h_0 \rightarrow h  } \text{MoM}(h,t) = \langle \hat L^2_{h} \rangle = F_{h}[\hat{\rho}_h].
    \end{equation}

\subsection{SLD-based MoM protocol}

    In order to explicitly write the estimator
    $\hat{h}_{\rm MoM}$ and the protocol for the MoM, using the SLD as its observable, we can first perform a Taylor expansion of \(\langle\hat{L}_{h_0}\rangle\) in $h$ around \(h_0\). Noting that according to definition in Eq.~\eqref{eq:03} at $h_0$ point we have that
    \begin{equation}
        \langle\hat{L}_{h_0}\rangle\big|_{h=h_0}=0,\qquad
\partial_h\langle\hat{L}_{h_0}\rangle\big|_{h=h_0}=\langle\hat{L}_{h_0}^2\rangle,
    \end{equation}
        one obtains that,
    \begin{equation}
        \langle\hat{L}_{h_0}\rangle\approx \langle\hat{L}_{h_0}^2\rangle (h-h_0)\quad\text{for }|h-h_0|\ll h.
    \end{equation}
    Therefore, the  MoM estimator for \(\nu\) independent measurements is given by,
    \begin{equation}
        \label{eq:SLD.estimator}
        \hat{h}_{\rm MoM}=h_0+\frac{\langle\hat{L}_{h_0}\rangle}{\langle\hat{L}_{h_0}^2\rangle}.
    \end{equation}

    As mentioned in Sec.~\ref{sec.mom}, the moments are estimated from the $i$-shot measurement outcomes \(m_i\) obtained by the POVM of the SLD operator,
	\begin{equation}
    	\langle\hat{L}_{h_0}\rangle\approx\frac1k \sum_{i=1}^k m_i,\quad
    	\langle\hat{L}_{h_0}^2\rangle\approx\frac1k \sum_{i=1}^k m_i^2.
	\end{equation}
	Therefore, from the above discussion the protocol can be defined as follows:

    (i)~Choose an initial value \(h_0\) (e.g., from a prior rough measurement, theoretical prediction, or calibration run).

    (ii)~Prepare the system in the initial state \(|\psi_0\rangle\) and evolve it for an interrogation time \(t^*\), typically determined by the system parameters, that maximizes the QFI.

    (iii)~At time \(t^*\), construct the operator \(\hat{L}_{h_0}\) from the pure-state SLD formula (Eq.~\eqref{eq:SLD.pure.states}) and perform its projective measurement, recording the outcome \(m_i\) for the \(i\)-th shot.

    (iv)~After \(k\) independent shots, apply the MoM estimator of Eq.~\eqref{eq:SLD.estimator} to obtain \(\hat{h}_{\rm MoM}\).

	In general, however, this protocol may be unrealistic for practical laboratory implementations since \(\hat{L}_{h}\) lacks a clear physical interpretation. To first grasp more intuition on it, we analyze in the next subsection a simple single-spin toy model, which shows how SLD operator evolve in terms of Pauli spin observables.

	\subsection{SLD trajectories on the Bloch sphere: single-spin model}

    \begin{figure}
        \centering
        \includegraphics[width=0.49\textwidth]{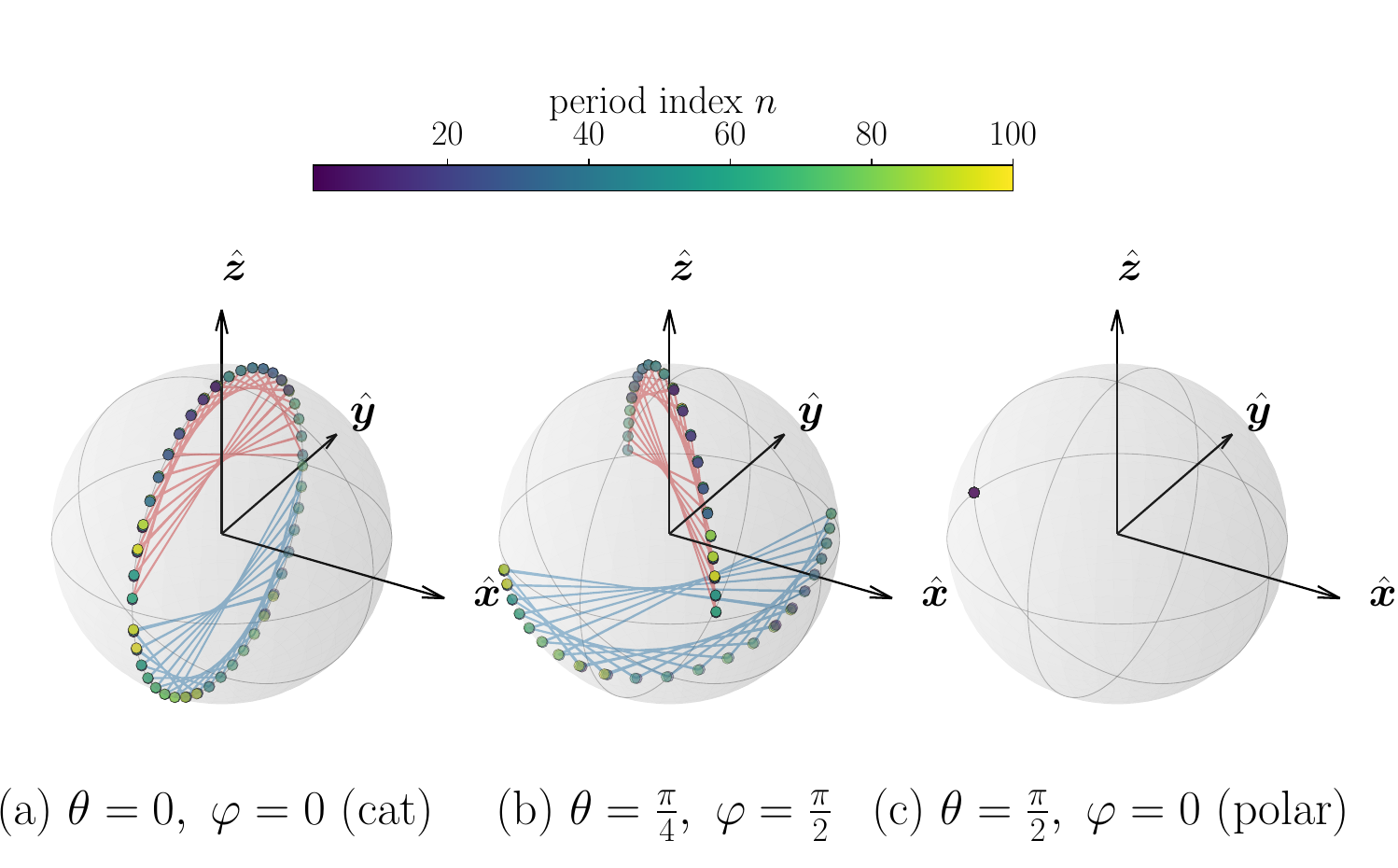}
        \caption{Bloch-sphere trajectories of the SLD vector $\vec{L}_h(nT)/|\vec{L}_h(nT)|$ for the kicked single-spin Hamiltonian---Eq.~\eqref{eq:single.spin.model}---in the linear response regime ($h\rightarrow 0$).
        Markers indicate integer values of the period index $n=1,\dots,100$ (color encodes $n$); thin curves connect even-$n$ (blue) and odd-$n$ (red) points separately to expose the period-doubling structure inherited from the kick.
        We consider three representative initial states, drawn from  $\ket{\psi_0(0)} = \cos(\theta/2)\ket{E_+} + e^{i\varphi}\sin(\theta/2)\ket{E_-}$. (a) Cat-like structure,  $\theta=\varphi=0$: with a  trajectory confined to the $y$--$z$ plane. (b) $\theta=\pi/4$, $\varphi=\pi/2$: full three-dimensional precession. (c) Polarized state $\theta=\pi/2$, $\varphi=0$: with $\hat{L}_h(nT)$ pinned to $-\hat{x}$ direction. In all cases we use the system parameters $\Delta=1$, $T=1$, and $\phi = \pi$.}
        \label{fig:sld-trajectories-lmg}
    \end{figure}

In order to illustrate the SLD operator physically we first consider a very simple model, described by a kicked single-spin system under a constant magnetic field and a transverse driven ac field with amplitude $h$. Specifically,
    \begin{equation}
    \label{eq:single.spin.model}
        \hat H_{ss}(h) = -\frac{\Delta}{2} \hat \sigma^x + h \theta(t) \hat \sigma^z - \frac{\phi}{2} \sum_{n=1}^\infty \delta(t-nT) \hat \sigma^x,
    \end{equation}
    where  $\Delta>0$ is the bare gap, the ac signal is given by a Heaviside step function in period-doubling resonance to the kicking,  i.e.
        $\theta(t+kT) = (-1)^k \theta(t)$ with $|\theta(t)|=1\, \forall t$, and $h$ is the desired parameter to be estimated.
    The eigenvectors of the bare hamiltonian $\hat{H}_{\text{ss}}(0)|_{t=0}=-(\Delta/2)\hat\sigma^x$, written in the computational basis $\hat{\sigma}^z$ are given by\cite{SM},
    \begin{equation}
    	\label{eq:basis.single.spin}
        \ket{E_\pm} = \frac{1}{\sqrt{2}}
        \begin{pmatrix}
            1 \\ \pm1
        \end{pmatrix},  \quad E_{\pm}=\mp\frac{\Delta}{2}.
    \end{equation}
    For such single spin model, one can compute analytically $\ket{\psi_h(t)}$ and $\ket{\partial_h \psi_h(t)}$, and obtain the SLD operator from Eq.~\eqref{eq:SLD.pure.states}. We show in  Fig.~\ref{fig:sld-trajectories-lmg} the resulting trajectories for the normalized SLD operator, defined by its decomposition on the Bloch sphere $\vec{L}_h(nT)\big|_\alpha = \tfrac{1}{2}\Tr[\hat{L}_h\hat{\sigma}^\alpha]$ for $\alpha=x,y,z$. While the general behavior (middle-panel) features a nontrivial dynamics, spreading all over the Bloch sphere, one can see that for certain initial states, the SLD vector remains confined to a plane and aligned with the observable ${\hat \sigma}^z$ at a specific time (left-panel) or aligned with the $\hat \sigma^x$ trajectory for all times (right-panel).

    In the many-body case, however, the dynamics of the SLD operator can become substantially more intricate, making the observable too complex. Specifically,  the SLD operator can evolve into a nonlocal observable, precluding a simple and direct realization in experiments. To mitigate this complexity,  we show below that for certain initial states the SLD operator can still assume a remarkably simple form.

    \section{SLD in Floquet Time Crystals}
    \label{sec:SLD in Floquet Time Crystal}

In this section, we derive the analytical expansion (and further simplifications) for the dynamics of the SLD operator in an FTC phase and discuss its main properties in connection to the MoM. Firstly, we demonstrate a general relation between the SLD operator and the HSO. This relation allows us to transfer many known properties of the HSO to the realm of the SLD operator. We then discuss the explicit structure of the SLD operator, highlighting restrictions that may prevent its use in a metrology MoM protocol from an experimental point of view.
Finally, to mitigate these restrictions, we study the analytical form of the SLD operator for a few relevant initial states, demonstrating that operators such as parity and global magnetization, which are feasible from an experimental point of view, could still be used instead.

    \subsection{Connection between SLD and HSO}
    It is first useful to show how the SLD is generally connected to the HSO of Eq.~\eqref{eq:HSO}. Considering the  unitary evolution $|\psi_h (t) \rangle = \hat U_h (t) |\psi_0 \rangle$,
     and a perturbation theory expansion of $\ket{\psi_h}$ around $h$, for small  $\delta h \ll h$, we have that
    \begin{eqnarray}
        \label{eq.psi.expansion}
        \begin{aligned}
            \ket{\psi_{h+\delta h}(t)} &= \ket{\psi_h^{(0)}(t)} + \delta h  \ket{\psi_h^{(1)}(t)} + O(\delta h^2)
        \end{aligned}
    \end{eqnarray}

    Alternatively, the evolution of the initial state
    $\ket{\psi_0}\equiv\ket{\psi_h^{(0)}(0)}$ with a perturbed dynamics $\hat U_{h+\delta h}(t)$ is given in the Schrödinger picture as,
    \begin{equation}
        \label{eq:taylor.psi.evolution}
        \begin{aligned}
            \ket{\psi_{h+\delta h}(t)} &= \hat{U}_{h+\delta h}(t)\ket{\psi_{h}(0)} \\
            &= \left(\hat{U}_h + \delta h  \frac{\partial \hat{U}_h}{\partial h} \right)\ket{\psi_h^{(0)}(0)} + O(\delta h^2) \\
            &= \hat{U}_h(t)\ket{\psi_h^{(0)}(0)} - i \delta h\hat{U}_h(t)\hat{S}_h(t)\ket{\psi_h^{(0)}(0)} \\
            &+ O(\delta h^2),
        \end{aligned}
    \end{equation}
    where in the second line we used a Taylor expansion of the unitary operator up to first order in $\delta h$, and in the third line we used Eq.~\eqref{eq.qfi_def} for the derivative.

    We thus  observe, comparing Eq.~\eqref{eq:taylor.psi.evolution} with Eq.~\eqref{eq.psi.expansion}, that
    \begin{eqnarray}
    \label{eqs.psi.firstorders}
        \begin{aligned}
            \ket{\psi_h^{(0)}(t)} &= \hat{U}_h(t)\ket{\psi_0}, \\
            \ket{\psi_h^{(1)}(t)} &= -i \hat{U}_h(t)\hat{S}_{h}(t) \ket{\psi_0}.
        \end{aligned}
    \end{eqnarray}

Using now the above relations to the SLD operator allow us to obtain its dependence to the HSO. Specifically, from the perturbation theory expansion (Eq.~\eqref{eq.psi.expansion}) and Eq.~\eqref{eq:SLD.pure.states} one has that,
    \begin{equation}
        \label{eq.SLD}
        \begin{aligned}
            \hat {L}_{h+\delta h} =& 2 (\ket{\psi_h^{(0)}(t)} \bra{\psi_h^{(1)}(t)} + \textrm{h.c.}) +  O(\delta h)
        \end{aligned}
    \end{equation}
    Employing that the first orders corrections are described by Eq.~\eqref{eqs.psi.firstorders}, we obtain the connection between HSO and SLD operators in the limit $\delta h \to 0$,
    \begin{equation}
        \label{SLD:general.form}
        \begin{aligned}
            \hat{L}_{h} (t) &= 2i  \hat{U}_h(t) \Big[\hat \rho_0 \,, \hat{S}_{h}(t) \Big] \hat{U}^{\dagger}_h(t)
        \end{aligned}
    \end{equation}
    where  the square brackets denotes the commutator, and  $\hat{\rho}_0 \equiv \ket{\psi_0}\bra{\psi_0}$ highlights the initial state dependence.

    \subsection{HSO in Floquet time crystals}

    Given the previous connection between SLD and HSO, it shall thus inherit many of its properties and structure. In fact, as we show, the SLD can be analyzed analytically using HSO using two main assumptions. First, the FTC sensor is assumed to be in contact with an ac field $\hat V_h(t) = h~\rm{sign}{\left(f(t)\right)} {\hat S}_z$,  with $f(t)$ a  periodic function which is in period-doubling resonance with the sensor, $f(t + T) = -f(t)$. The sign function is defined as,
    \begin{equation}
        \mathrm{sign}(x)=
            \begin{cases}
                +1, & x\ge0,\\
                -1, & x<0.
            \end{cases}
    \end{equation}
    Second, throughout the discussion below we consider the linear-response regime with \(h\to0\). These assumptions can nevertheless be generalized in a straightforward manner: the nonlinear regime can be treated as in Ref.~\cite{Andrei2026}, using an effective Floquet Hamiltonian. The derivations presented below therefore remain valid in the general case, up to minor modifications, but we adopt the assumptions above for simplicity.

    Under linear response regime, the sensor is described by the Floquet unitary,
    \begin{equation}
        \hat{U}_{h\to0}(nT) = \left( \hat{X} e^{-i \hat H_F T} \right)^n
    \end{equation}
    where $\hat X$ is the kick operator and $\hat H_F$ is the Floquet Hamiltonian. These operators satisfy the relations $\hat X^2 = \mathbb{I}$, $[\hat X,\hat H_F] = 0$, and $\{\hat X, \hat S_z\} = 0$.

    Due to the commutation relation, the kick and Floquet Hamiltonian share the same set of eigenstates, with the following spectral properties,
    \begin{equation}
        \label{eq:hf-floquet-schrodinger}
\hat{H}_F|E_{i}\rangle=E_{i}|E_{i}\rangle,\qquad\hat{X}|E_{i}\rangle=p_i|E_{i}\rangle,
    \end{equation}
    where \(p_i=\pm1\) is the ``parity'' of the \(i\)'th eigenstate. Using the above relations and some algebraic manipulations (see \cite{Andrei2026}), one can derive the HSO in the eigenbasis of the Floquet Hamiltonian,
    \begin{equation}
        \label{eq:Sh.expansion}
        \hat{S}_{h\to0}(nT)=\sum_{i,j}\mathcal{O}_{ij}R_{ij}(nT)|E_{i}\rangle\langle E_{j}|,
    \end{equation}
    where \(\mathcal{O}_{ij}=\langle E_{i}|\mathcal{\hat{S}}_z|E_{j}\rangle\) and the response term \(R_{ij}(nT)\) for sign function is
    \begin{equation}
        \label{eq:response.term}
        R_{ij}(nT)=\frac{\sin(\Delta_{ij}nT/2)}{\Delta_{ij}/2}\,e^{i\Delta_{ij}nT/2},
    \end{equation}
    with \(\Delta_{ij}=E_{i}-E_{j}\).

    As a result, since the energy gap of a cat subspace is exponentially small, its response terms give the dominating contribution to the HSO. Such a cat-paired subspace is defined by pairs of quasi-degenerate eigenstates with opposite parities for finite system sizes, represented by \(\{|E_{i}\rangle,|E_{\bar{i}}\rangle\}_{i=1}^M\), where \(M\le d_H\) with \(d_H\) the Hilbert-space dimension and the index \((i,\bar{i})\) denotes the \(i\)'th paired states. Moreover, under period-doubling resonance the response terms \(R_{ij}(t)\) dephase for times larger than their inverse gap (\(t \gtrsim t_{ij}^* \equiv \Delta^{-1}_{ij}\)). Therefore, after an initial transient time \(nT \sim O(\Delta_{i,j\neq i,\bar{i}}^{-1})\), the HSO reduces to a block-diagonal form along the cat-paired subspaces\cite{Andrei2026},
    \begin{equation}
        \label{eq:hso-block-diagonal}
        \begin{aligned}
            &\hat{S}_{h\to0}(nT) \approx   \hat{S}_{\mathrm{bd}}(nT)\\
            & =
        \begin{pmatrix}
            \hat{s}^{[i_1\bar{i}_1]}(nT) & 0      & \hdots                      & 0                     \\
            0                           & \ddots &                             & \vdots                \\
            \vdots                      &        & \hat{s}^{[i_M\bar{i}_M]}(nT) & 0                     \\
            0                           & \hdots & 0                           & \hat{0}^{[d \times d]}
        \end{pmatrix}.
        \end{aligned}
    \end{equation}
    Here, the approximation neglects terms that are sublinear in time, and \(\hat{s}^{[i\bar{i}]}(nT) \equiv (\hat{S}_{h\to0}(nT))_{m,n=(i,\bar{i})}\) are the block-diagonal terms. Specifically, these are \(2\times2\) matrices whose elements are those of the paired cat states,
    \begin{equation}
        \hat{s}^{[i\bar{i}]}(nT) =
        \begin{pmatrix}
            0 & \mathcal{O}_{i\bar{i}}R_{i\bar{i}}(nT) \\
            \mathcal{O}_{\bar{i}i}R_{\bar{i}i}(nT) & 0
        \end{pmatrix},
    \end{equation}
    and \(\hat{0}^{[d\times d]}\) is the null matrix of size \(d\times d\), with \(d = d_H - 2M\) the dimension of the remaining spectrum. Substituting this block-diagonal ansatz for \(\hat{S}_{h\to0}\) into Eq.~\eqref{SLD:general.form} preserves its diagonal part, provided all operators share a common block structure. It is therefore sufficient to analyze a single \(2\times2\) cat-subspace block before addressing the general structure.

    \subsection{SLD in a single cat subspace}

    In the simpler, but nontrivial, case we consider an initial state within a single paired cat subspace, as follows,
    \begin{equation}
        \label{eq:init-psi}
        \ket{\psi^{[i]}_0} = \cos\left(\tfrac{\theta_i}{2}\right)\ket{E_{i}} + e^{i\varphi_i} \sin\left(\tfrac{\theta_i}{2}\right)\ket{E_{\bar{i}}},
    \end{equation}
    where $\ket{E_i}$ and $\ket{E_{\bar{i}}}$ denotes quasi-degenerate eigenstates of opposite parity $p_{\bar{i}} = -p_i$ for finite system sizes.
    In this case the stroboscopic evolution under $\hat{U}_{h\to0}(nT)$ is given by
    \begin{equation}
        \hat{U}_{h\to0}(nT)\ket{E_{i}} = p_i^n e^{-i E_{i} nT} \ket{E_{i}},
    \end{equation}
    where $p_i = \pm 1$ encodes the parity of $\ket{E_{i}}$. Therefore, substituting the SLD expansion Eq.~\eqref{eq:Sh.expansion}, response terms Eq.~\eqref{eq:response.term} and initial state Eq.~\eqref{eq:init-psi} into Eq.~\eqref{SLD:general.form}, and moreover using the above relation with the properties of $\hat{\mathcal{S}_z}$ in the cat subspace ($\textrm{Im} (\mathcal{O}_{i\bar{i}})=0$)~\cite{SM}, one obtains:
    \begin{equation}
        \begin{aligned}
            &\hat{L}_{h\to0}(nT) =  \\
            \mathcal{O}_{i\bar{i}} &\bigg(
            \mathcal{X}_{i}(nT)\ket{E_{i}}\bra{E_{i}} + \mathcal{Z}^*_{i}(nT) \ket{E_{i}}\bra{E_{\bar{i}}} \\
            & + \mathcal{Z}_{i}(nT) \ket{E_{\bar{i}}}\bra{E_{i}} - \mathcal{X}_{i}(nT)\ket{E_{\bar{i}}}\bra{E_{\bar{i}}} \bigg)
        \end{aligned}
    \end{equation}
    where the coefficients are given by,
    \begin{equation}
        \label{eq:coefficients-bd}
        \begin{aligned}
            \mathcal{X}_{i}(nT)&\equiv \frac{2 \sin{\left(\Delta_{i\bar{i}}  nT / 2\right)}}{\Delta_{i\bar{i}}/2}\sin (\theta_i ) \sin\left(\Delta_{i\bar{i}}  nT/2+\varphi_i \right), \\
            \mathcal{Z}_{i}(nT)&\equiv i (p_i p_{\bar{i}})^n \frac{2 \sin{\left(\Delta_{i\bar{i}}  nT / 2\right)}}{\Delta_{i\bar{i}}/2}  \cos(\theta_i) e^{-i\Delta_{i\bar{i}}  nT/2}.
        \end{aligned}
    \end{equation}
    Here, we refrain from simplifying the case $p_i p_{\bar{i}}=-1$ in order to preserve the general structure of the expressions. One may verify that these coefficients reproduce the analytical expansion for $\hat L_{h\to0}$ obtained for the toy model in Eq.~\eqref{eq:single.spin.model} for the single-spin case.\cite{SM} The coefficients have dominant prefactors $\Delta^{-1}_{i\bar{i}} \propto e^N$ for FTC, which allows for further simplification of the analysis.

    \begin{table*}
        \centering
        \begin{tabular}{l c c c}
            \hline
            \textbf{Initial state} & \textbf{Optimal observable} & \textbf{Time window} & \textbf{Dynamics term} \\
            \noalign{\vskip 2mm}
            \hline
            Polarized ($\theta_i = \pi/2$, $\varphi_i = 0$) & Parity $\hat{X}$ & $t \gtrsim O(1)$ & $\mathcal{X}_{i}(nT)$ \\
             Cat state ($\theta_i = \pi/2$, $\varphi_i \neq 0$)   & Parity $\hat{X}$ & $t \gtrsim O(1)$ & $\mathcal{X}_{i}(nT)$ \\
            Cat state ($\theta_i = 0$, $\varphi_i = 0$) & $\hat{\mathcal{Y}} = i \hat{\mathcal{S}}_z \hat X$ & $t \lesssim \Delta_{i\bar{i}}^{-1}$ & $\Im\{\mathcal{Z}_{i}(nT)\}$ \\
            Cat state ($\theta_i = 0$, $\varphi_i = 0$) & Signal $\hat{\mathcal{S}}_z$ & $t \simeq \Delta_{i\bar{i}}^{-1}$ & $\Re\{\mathcal{Z}_{i}(nT)\}$ \\
            \noalign{\vskip 2mm}
            \hline
        \end{tabular}
        \caption{Optimal observables read off from the structure of $\hat{L}_{h\to0}$ for different initial states (see Eq.~\eqref{eq:SLDbd-initial-state}, where $\{\theta_i,\varphi_i\}$ are considered for simplicity identical for all $1 \le i \le M$), together with the time window over which each observable is optimal and the corresponding dynamical term (Eq.~\eqref{eq:coefficients-bd}). Here $\Delta_{i\bar{i}}^{-1}$ denotes the time inversely proportional to the cat-subspace gap $\Delta_{i\bar{i}}$. The condition $t \lesssim \Delta_{i\bar{i}}^{-1}$ indicates that the maximum overlap with the optimal observable is reached before the characteristic cat-subspace time, whereas $t \simeq \Delta_{i\bar{i}}^{-1}$ means that it is reached only around that time.  The condition $t\gtrsim O(1)$ indicates that this occurs after the initial transient time, required for the emergence of the block-diagonal structure in the SLD operator.
        }
        \label{tab:summary.observables}
    \end{table*}

    \subsection{Block-diagonal structure of the SLD}

    We consider  a general initial state in the form,
    \begin{equation}
        \label{eq:SLDbd-initial-state}
        |\psi_0\rangle = \sum_i c_i |\psi_0^{[i]}\rangle + c_\chi |\chi \rangle,
    \end{equation}
    where $\sum_{i=1}^M |c_i|^2 + |c_{\chi}|^2 = 1$, $|\psi_0^{[i]}\rangle$ are defined by Eq.~\eqref{eq:init-psi} and $|\chi \rangle$ is a state orthogonal to the cat subspaces. In density matrix form, this initial state can be written as,
    \begin{equation}
    	\label{eq:rho-general}
            \hat{\rho}_0 = \hat \rho_{\textrm{bd}} + \hat \rho_{\textrm{off}}
    \end{equation}
     where we explicitly decompose it in terms of its block diagonal terms along the cat subspaces,
     \begin{eqnarray}
         \label{eq:rho-diagonal-part}
        \hat \rho_{\textrm{bd}} &=&
            \sum_{i=1}^M |c_i|^2 \ket{\psi^{[i]}_0} \bra{\psi^{[i]}_0},
    \end{eqnarray}
    and those off-diagonal and orthogonal to it,
    \begin{eqnarray}
            \hat \rho_{\textrm{off}} &=&
             \sum_{i\ne j}^M c_i c^*_j \ket{\psi^{[i]}_0} \bra{\psi^{[j]}_0}
            + \sum_{i=1}^M  \left( c_i c^*_{\chi}\ket{\psi^{[i]}_0} \bra{\chi} + \text{h.c.} \right) \nonumber \\
            & & \qquad + |c_{\chi}|^2 \ket{\chi}\bra{\chi}.
    \end{eqnarray}

We can decompose the contribution of these two initial state terms in the SLD as given by Eq.~\eqref{SLD:general.form}, i.e.
 \begin{equation}
        \hat{L}_{h\to0} = \hat{L}_{\textrm{bd}} + \hat{L}_{\textrm{off}}
    \end{equation}
    where
    \begin{eqnarray}
        \hat{L}_{\textrm{bd}} &=& 2i  \hat{U}_{h\to0}(nT) \Big[\hat \rho_{\textrm{bd}} \,, \hat{S}_{h\to0}(nT) \Big] \hat{U}^{\dagger}_{h\to0}(nT), \nonumber \\
        \hat{L}_{\textrm{off}} &=& 2i  \hat{U}_{h\to0}(nT) \Big[\hat \rho_{\textrm{off}} \,, \hat{S}_{h\to0}(nT) \Big] \hat{U}^{\dagger}_{h\to0}(nT). \nonumber \\
    \end{eqnarray}

An important observation is now in order. As previously mentioned, after the initial transient time the HSO operator $\hat S_{h\to0}$ takes a block-diagonal form (Eq.~\eqref{eq:hso-block-diagonal}).
Consequently, the diagonal or off-diagonal structure of the initial state  are inherited to their corresponding SLD decomposition terms.
In particular, the block diagonal SLD term reduces to,
    \begin{equation}
        \label{eq.heisenb.signal.blockdiagonal}
        \begin{aligned}
            \hat{L}_{\textrm{bd}}
            & \approx
            \begin{pmatrix}
                \hat{l}^{[1\bar{1}]}(nT) &  0      & \hdots                      & 0      \\
                0                         &  \ddots   &                             & \vdots \\
                \vdots                    &           & \hat{l}^{[M\bar{M}]}(nT) & 0      \\
                0                         &  \hdots   &    0                        & \hat{0}^{[d \times d]}
            \end{pmatrix}
        \end{aligned}
    \end{equation}
    which is  therefore limited to cat-subspaces after the initial transition time.  Each of such blocks are given by,
    \begin{equation}
        \begin{aligned}
            \hat{l}^{[i\bar{i}]} (nT) = \mathcal{O}_{i\bar{i}} |c_i|^2
            \begin{pmatrix}
                 \mathcal{X}_{i}(nT) & \mathcal{Z}^*_{i}(nT) \\
                 \mathcal{Z}_{i}(nT) & -\mathcal{X}_{i}(nT)
            \end{pmatrix},
        \end{aligned}
    \end{equation}
    with same coefficients as Eq.~\eqref{eq:coefficients-bd}.

    Given such a structure we can compute the similarity of the SLD operator to a few relevant observables. Specifically, its similarity to the parity, magnetization or a dressed parity-magnetization observable, defined as $\hat{\mathcal{Y}} \equiv i \hat{\mathcal{S}}_z \hat X$, which can be expanded as follows:
\begin{equation}
        \begin{aligned}
            \hat X &= \sum_{i=1}^M p_i \left(\ket{E_i} \bra{E_i} - \ket{E_{\bar{i}}} \bra{E_{\bar{i}}} \right) + { X }_\chi \ket{\chi} \bra{\chi}, \\
                            \hat{\mathcal{Y}}
            &\approx \sum_{i=1}^M  i \mathcal{O}_{i\bar{i}} p_i\left(\ket{E_{\bar{i}}} \bra{E_{i}} - \ket{E_i} \bra{E_{\bar{i}}}\right) + i \mathcal{O}_{\chi} { X }_\chi \ket{\chi} \bra{\chi}, \\
            \hat{\mathcal{S}}_z &\approx \sum_{i=1}^M \mathcal{O}_{i\bar{i}} \left(\ket{E_i} \bra{E_{\bar{i}}} + \ket{E_{\bar{i}}} \bra{E_{i}}\right) + \mathcal{O}_{\chi}\ket{\chi} \bra{\chi}.
        \end{aligned}
    \end{equation}

    In the above expansion we neglect terms of the $\hat{\mathcal{S}}_z$ operator which are non-diagonal in the cat-subspaces.  Although this constitutes a simplification, it is grounded in the microscopic structure of cat-like states, where the dominant terms reside within the same subspace. To corroborate this intuition, we analytically and numerically evaluated this approximation for an FTC based on the Lipkin-Meshkov-Glick (LMG) model~\cite{russomanno_floquet_2017}. Our findings confirm this picture---see the analysis in the Supplemental Material~\cite{SM}.

    In this way, since the observables considered above are all block-diagonal in the cat-subspaces, their similarity to the SLD stems entirely from the $\hat{L}_{\textrm{bd}}$ term. Consequently, the overlap of the SLD with respect to these relevant observables approximates to,
    \begin{equation}
        \label{eq:SLD-projections}
        \begin{aligned}
            \mathrm{Tr}\{\hat X\, \hat{L}_{h\to0}(nT)\} &\approx \sum_{i=1}^M |c_i|^2 p_i\, \mathcal{O}_{i\bar{i}}\, \mathcal{X}_{i}(nT), \\
                            \mathrm{Tr}\{\hat{\mathcal{Y}}\, \hat{L}_{h\to0}(nT)\} &
            \approx \sum_{i=1}^M |c_i|^2 p_i \mathcal{O}_{i\bar{i}}^2\, \Im\{\mathcal{Z}_{i}(nT)\}, \\
            \mathrm{Tr}\{\hat{\mathcal{S}}_z\, \hat{L}_{h\to0}(nT)\} &\approx \sum_{i=1}^M |c_i|^2 \mathcal{O}_{i\bar{i}}^2\, \Re\{\mathcal{Z}_{i}(nT)\}.
        \end{aligned}
    \end{equation}

The above expressions reveal interesting simplifying structures -- as summarized in Table~\ref{tab:summary.observables}.
For a polarized initial state $\theta_i = \pi/2$, $\varphi_i = 0$, one has that $\mathcal{Z}_{i} = 0$ for all $0 \le i \le M$, therefore negligible overlap to $\hat{\mathcal{Y}}$ or $\mathcal{S}_z$ operators. By contrast, the parity operator has a non-negligible overlap to the SLD operator. Consequently, the optimal observable is expected to align with it. We recall that the scaling of this overlap corresponding to the term $\mathcal{O}_{i\bar{i}}$ -- in this case linearly with $N$ -- is immaterial since, as discussed, the MoM is invariant under a rescaling of the observable. Moreover, while a polarized state is usually associated with the case $\varphi_i=0$, a correlated cat-state structure would emerge for $\varphi_i \neq 0$. Nevertheless, the same reasoning also follow in this case, i.e. with $\theta_i = \pi/2$, but $\varphi_i \neq 0$, showing that with a simple parity observable, one could  saturate the QFI bound within the MoM given the preparation of such a state.

 In the case of different cat states with $\theta_i = \varphi_i = 0$, the parity coefficients now satisfy $\mathcal{X}_{i} = 0$, while the non-negligible terms arise along the
 $\hat{\mathcal{Y}}$ or $\mathcal{S}_z$ observables.
Therefore, the SLD operator is expected to be aligned with these last two observables. Notice that for times smaller than the FTC lifetime, $t \lesssim \Delta_{i \bar{i}}^{-1}$, the imaginary term of $\mathcal{Z}_i(t)$ are dominant, and in this way the SLD in this regime is dominantly a dressed parity-magnetization observable. On the other hand, for longer times, $t \simeq \Delta_{i \bar{i}}^{-1}$, the real part of $\mathcal{Z}_i(t)$ prevails, turning the SLD operators dominantly a bare magnetization observable.

\section{NMR platform} \label{sec:NMR platform}

    In this section we apply the theoretical framework developed previously to a specific physical system aiming to evaluate the MoM performance in estimating the unknown parameter $h$. Precisely, by: (i) evaluating the MoM estimation protocol in light of the SLD-based theoretical limit; (ii) implementing this analysis in the context of NMR, a platform widely employed in quantum metrology -- though rarely explored within the context of FTC phases; and (iii) bridging the gap between abstract theoretical derivation and practical realization. By utilizing realistic experimental parameters, we move beyond idealized numerical simulations, ensuring our analysis remains grounded in the physical constraints of contemporary NMR setups. In doing so, we address the critical challenges imposed by the finite lifetime of these FTC phases, specifically assessing how this temporal instability constrains the optimal measurement window for the MoM protocol. The physical model based on acetonitrile is presented in the next section.

    \subsection{NMR sensor}
    We consider a star-topology NMR system ~\cite{Mahesh_2021}, structurally analogous to molecules such as acetonitrile (CH$_3 $CN)~\cite{pal2018temporal}. The model---shown schematically  in Fig.~\ref{fig:model}---is particularly well-suited for experimental implementation~\cite{Uhlig2019}.

    \begin{figure}
        \centering
        \includegraphics[width=0.8\linewidth]{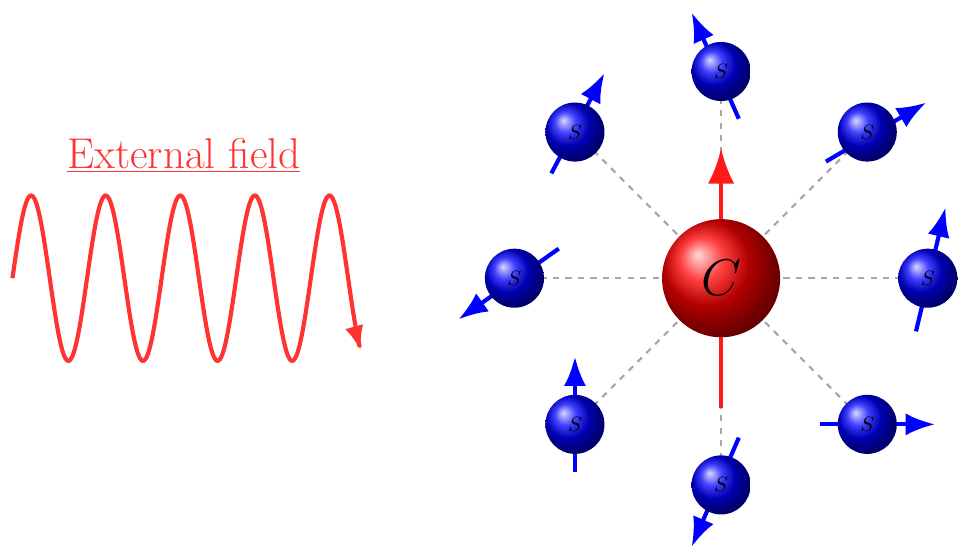}
        \caption{Conceptual diagram of the star-topology model~\cite{Mahesh_2021} for an NMR-based quantum sensor. The central spin (\textbf{C}) is coupled via $J$ to peripheral spins (\textbf{S}). The sensor interacts with an external ac field $h$, the amplitude of which is the parameter to be estimated.}
        \label{fig:model}
    \end{figure}

    The NMR spin system can be modeled with a star topology symmetry described by the Hamiltonian,

    \begin{equation}
    	\label{eq:nmr1}
    	\hat H(t) = \hat H_{\rm Is} + \hat H_{\rm ds} + \hat H_{\rm kick}(t) + \hat H_{\rm ac}(t),
    \end{equation}
    with,
    \begin{eqnarray}
    	\hat H_{\rm Is} &=& - (J/4) \hat \sigma_0^z  \sum^{N-1}_{i=1} \hat \sigma^z_i, \nonumber \\
        \hat H_{\rm ds} &=& \sum^{N-1}_{i=0} (h^x_i \hat \sigma^x_i + h^z_i \hat \sigma^z_i)/2, \nonumber \\
        \hat H_{\rm kick}(t) &=& (\vartheta - \epsilon_{\vartheta})\sum_{n=0}^\infty \delta(t-nT) \hat{M}_x, \nonumber \\
        \hat H_{\rm ac}(t) &=& hf(t)\hat{M}_z.
    \end{eqnarray}
        Here, $\hat \sigma_i^{\alpha=x,y,z}$ denotes the Pauli operators acting on the $i$'th spin, and $\hat{M}_\alpha = \sum_{i=0}^{N-1} \hat \sigma_i^\alpha/2$ represents the total spin magnetization.
        The term $ \hat H_{Is}$ describes an Ising-type interaction of strength $J$ between the central spin $(\hat \sigma^z_0)$ and each of $N-1$ peripheral spins. The second term,  $\hat H_{ds}$, accounts for disordered static local fields, consisting of longitudinal $(h_i^z)$ and transverse $(h_i^x)$ components. The disordered terms have strength $h_i^z \in [0, \, h^z]$,
        $h_i^x \in [-h^x, \, h^x]$,  chosen from a uniform distribution in the range of $h^z$,  $h^x$. The kick term $\hat H_{\textrm{kick}}$ represents a sequence of periodic $\vartheta-\epsilon_{\vartheta}$ rotations around the $x$ axis applied at intervals $T$, modeled as delta-function pulses. Finally, $\hat{H}_{ac}$ describes the external signal to be estimated, given by a global, sinusoidal driving field $f(t) = \sin(\omega_h t)$ with amplitude $h$ and frequency $\omega_h$.

    \begin{figure*}
        \centering
        \includegraphics[width=0.325\linewidth]{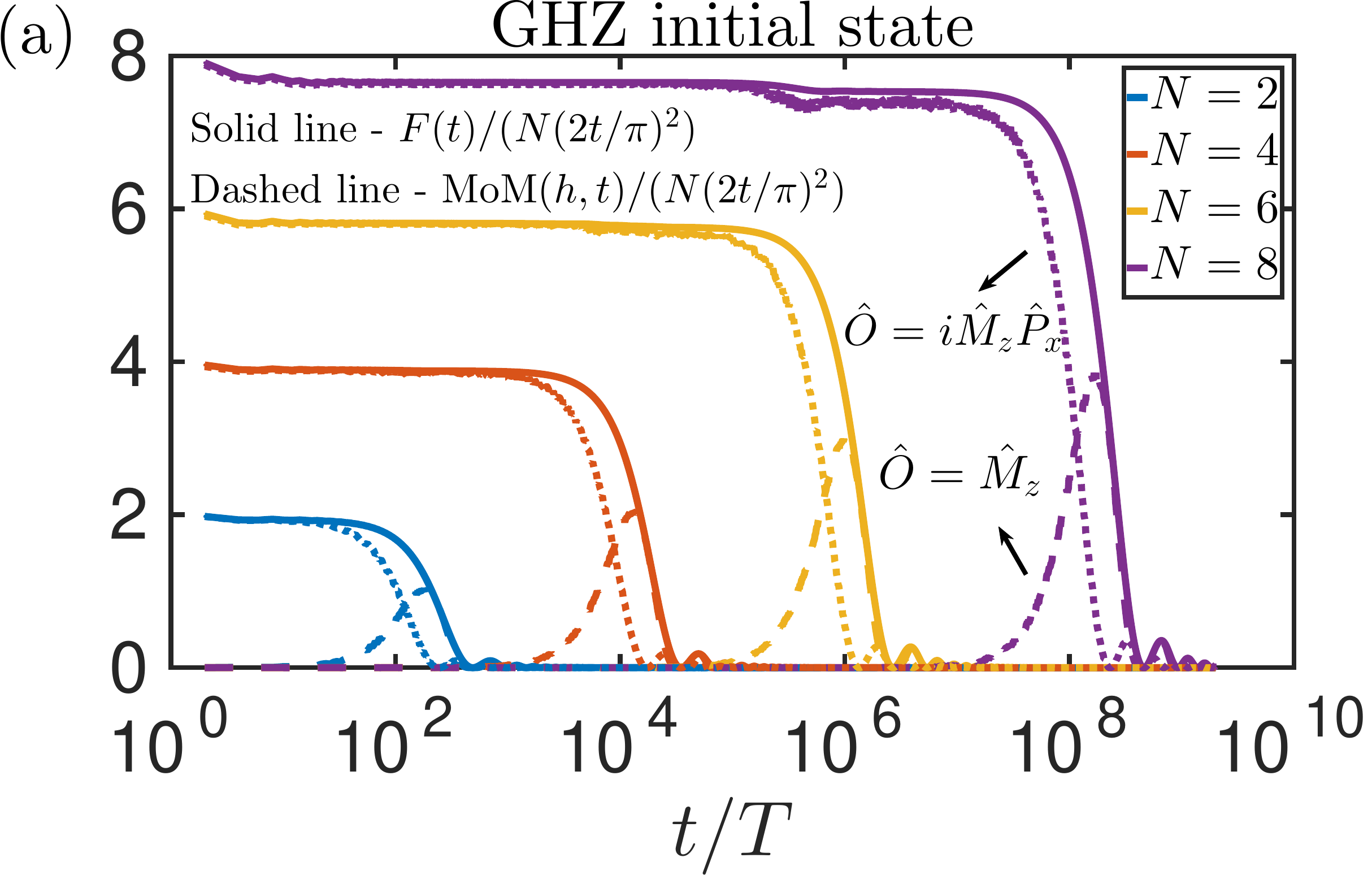}
        \includegraphics[width=0.325\linewidth]{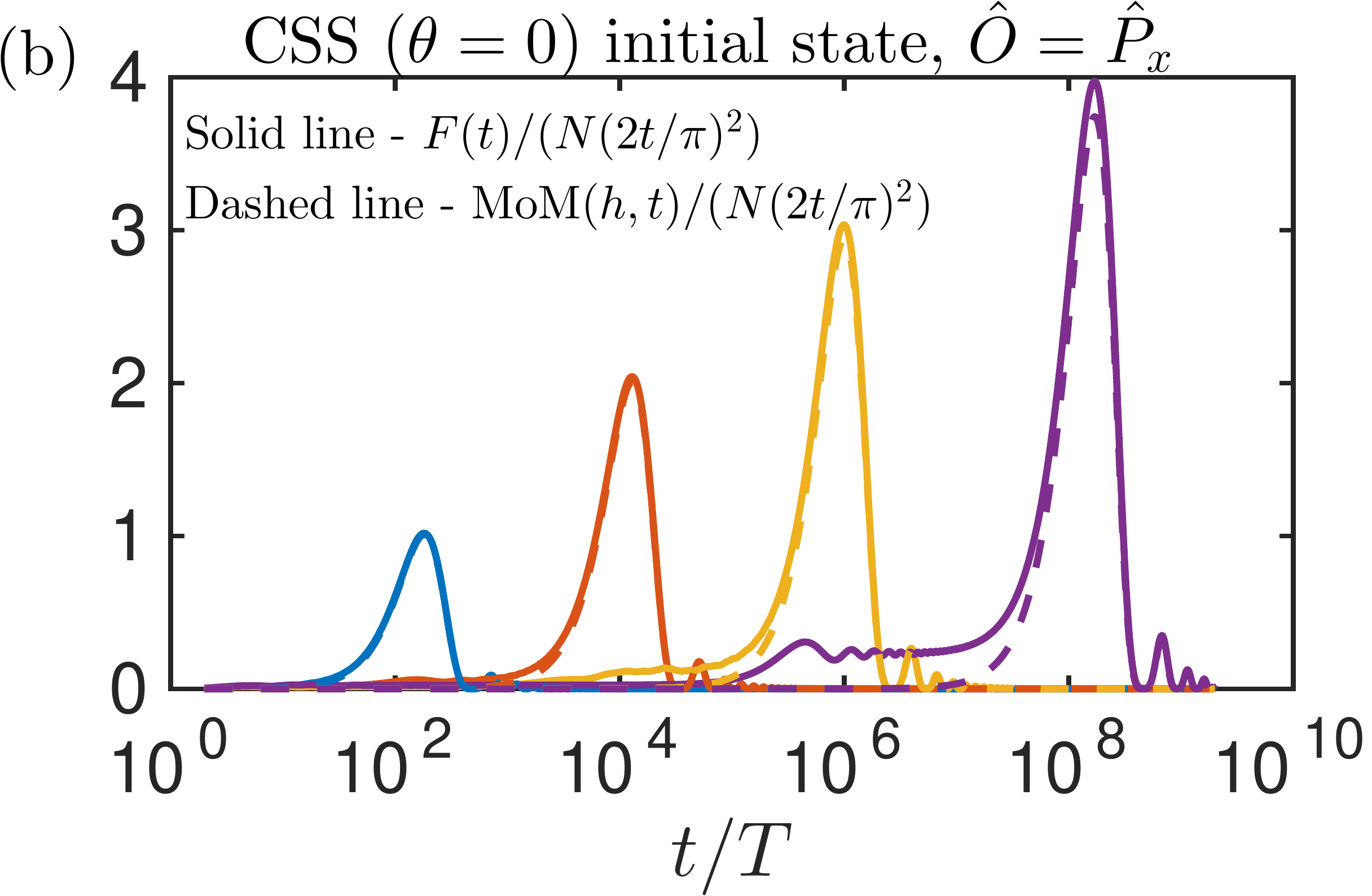}
        \includegraphics[width=0.33\linewidth]{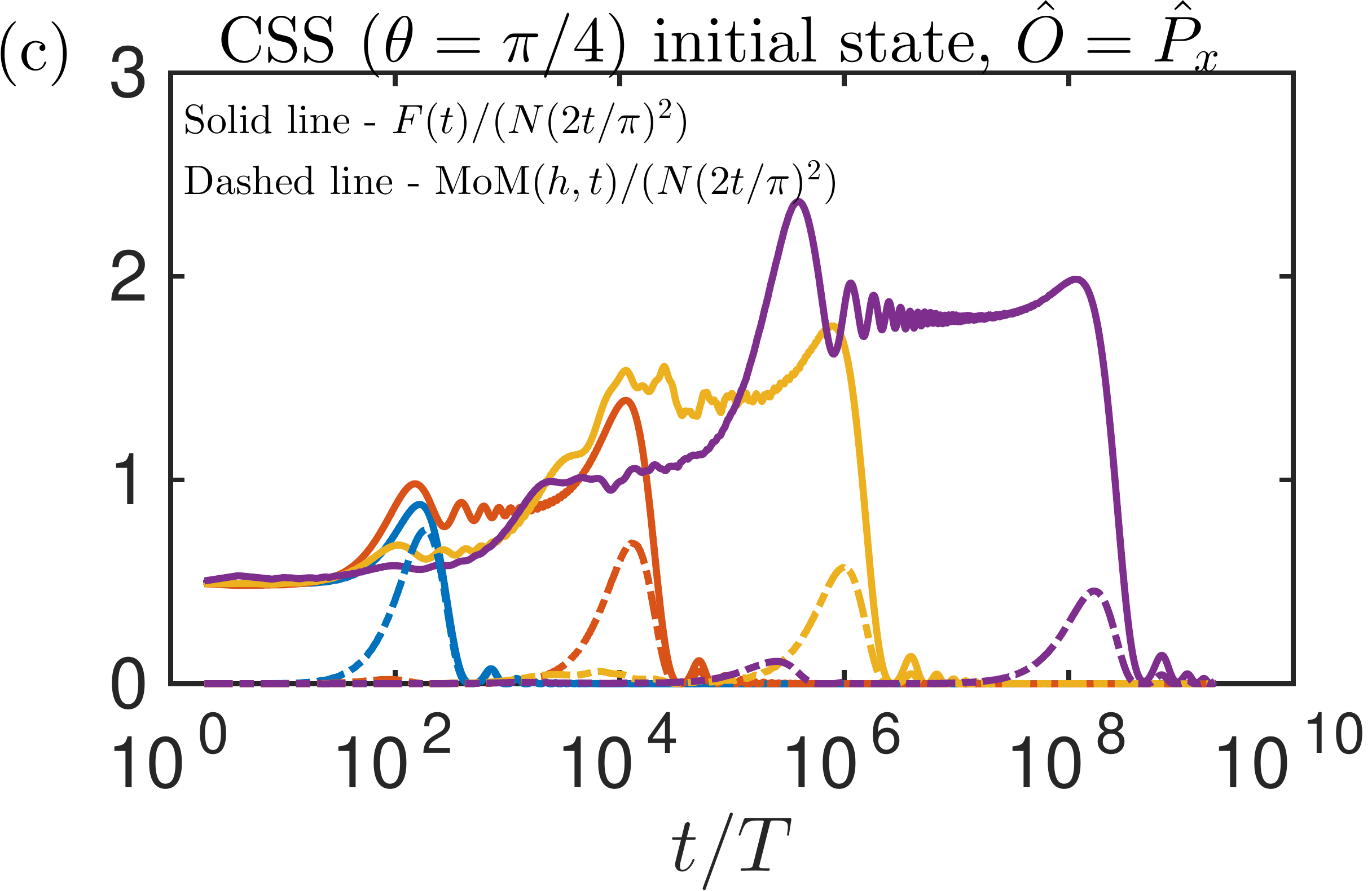}
        \caption{Dynamical behavior of the QFI (solid line) and MoM (dashed line) estimator in the linear response limit $(h \rightarrow 0)$, for   different system sizes $N$ and initial state preparations. (a) GHZ initial state with either the bare magnetization ($\hat M_{z}$) or a dressed parity-magnetization ($i\hat M_{z}\hat P_x$) observable. The QFI exhibits an initial plateau scaling quadratically with $N$, followed by a sudden collapse at the FTC lifetime. This behavior is  either closely saturated by MoM using the dressed observable, or tracked by a sharp peak at its lifetime by the bare magnetization observable. (b) CSS initial state with fully aligned spins $\theta = 0$, $\phi = 0$, or (c) misaligned ones $\theta = \pi/4$, $\phi = 0$. In both cases the MoM corresponds to the parity $\hat{P}_{x}$ as the observable. The dynamics reveal an initial $N$-independent regime, followed by non-monotonic growth before the final suppression.
        }
        \label{fig:3}
    \end{figure*}

    The phenomenology of this model is characterized by robust period-doubling dynamics of the magnetization, which persists for exponentially long times in the system size $N$-a hallmark feature of the FTC phase~\cite{PhysRevLett.120.180603, Choi2017, pal2018temporal}. This long-lived temporal order serves as a robust metrological resource, providing a stable periodic background against which external perturbations $h$ can be resolved with high precision~\cite{Moon2026, Andrmoon2024discretetimecrystalsensing}.

    Building upon this physical foundation, we now examine the metrological utility of these NMR dynamics for different initial states. Specifically, we analyze the MoM performance and demonstrate how the choice of observable affects the estimation, showing that it is possible to select good observables that remain experimentally feasible.

    \subsection{Dependence on the initial state and optimal observable}
    In laboratory implementations preparing the system’s initial quantum state is a critical step because it directly dictates the optimal observable for the metrology protocol. Significant experimental effort are typically devoted to preparing physical systems in either highly entangled Greenberger-Horne-Zeilinger (GHZ) states~\cite{GHZLI2025116257, GHZBugalho2025privaterobuststates} or coherent spin states (CSS)~\cite{CSS11396347, CSSChai_2025}. These experimental milestones provide the basis for our study on how the MoM protocol's performance is tied to the chosen initial state configuration. Specifically,  we consider the two paradigmatic cases:

    (i) GHZ initial state:
    \begin{equation}
        \ket{\rm{GHZ}} = \frac{\left( \ket{\uparrow}^{\otimes N } + \ket{\downarrow}^{\otimes N}  \right)}{\sqrt{2}},
    \end{equation}
which is characterized by its high sensitivity, making it an ideal candidate for detecting weak fields with high precision.

    (ii) CSS initial state:
    \begin{equation}
        |\rm{CSS}\rangle = \bigotimes^N_{i=1} \left( \cos(\theta/2)\ket{\uparrow}_i + e^{i\phi}\sin(\theta/2)\ket{\downarrow}_i \right),
    \end{equation}
where $\theta \in [0 \, \pi]$ and $\phi \in [0\,  2\pi]$,  which are typically employed in systems with a large number of particles with all spins polarized along the same direction.

   Following our previous discussion on optimal observables in the MoM estimator (Table~\ref{tab:summary.observables}), we associate the analysis of the total magnetization ($\hat M_z$) or a dressed parity-magnetization ($i\hat M_z \hat P_x$) given a GHZ initial state,  and the parity
\begin{equation}
        \hat P_x = \bigotimes^{N-1}_{i=0} \sigma_i^x \label{eq:px},
    \end{equation}
with the CSS initial configuration. 
We remark that, in general, the bare definitions of magnetisation $\hat M_z$ and parity $ \hat P_x$ introduced here do not necessarily coincide exactly with those discussed in Sec.~\ref{sec:SLD in Floquet Time Crystal}, which are derived from the effective Floquet Hamiltonian decomposition of the dynamics — a procedure that is typically nontrivial. Nonetheless, on physical grounds, one expects a significant overlap between these operators. 
Unless explicitly mentioned, we use in all our numerical simulations  an Ising coupling $J=2 $, a pulse of $ \vartheta = \pi$ with deviation of $\epsilon_{\vartheta} = 0.12$, disorder strength  $h^x = h^z = 0.01$, and a Floquet period $T=1$ which is a subharmonic frequency of the  ac field $\omega_h  = \frac{\pi}{T}$. We show the dynamics for a single  disorder realization, which is nevertheless representative of the general behavior.

 We show in Fig.~\ref{fig:3} our results for the dynamics of the QFI and MoM estimator. In Fig.~\ref{fig:3}(a) we see that, focusing first on the GHZ initial state, the QFI exhibits a high and stable plateau, persisting over timescales exponentially with the system size. Moreover, the plateau reaches the Heisenberg limit, where $F(t, h) \propto t^2 N^2$, due to a coherent collective dynamics and high entanglement among the spins in the system. This collective response ensures that the sensor remains highly sensitive as long as the system stays within the FTC phase. On the other hand, the MoM protocol (dashed lines) strongly depends on the measured observable. While for the dressed parity-magnetization observable it saturates the QFI till roughly the FTC lifetime, for the bare magnetization it is small most of the time showing though a localized peak that coincides with the onset of the QFI decay. Therefore, despite the simplicity of the bare magnetization observable, it can effectively concentrate the accumulated information about $h$ within a specific measurement window and serves as a practical tool for retrieving it, achieving maximum efficiency just before the finite lifetime of the FTC phase leads to eventual loss of sensitivity.

\begin{figure}
    \centering
    \includegraphics[width=0.5\linewidth]{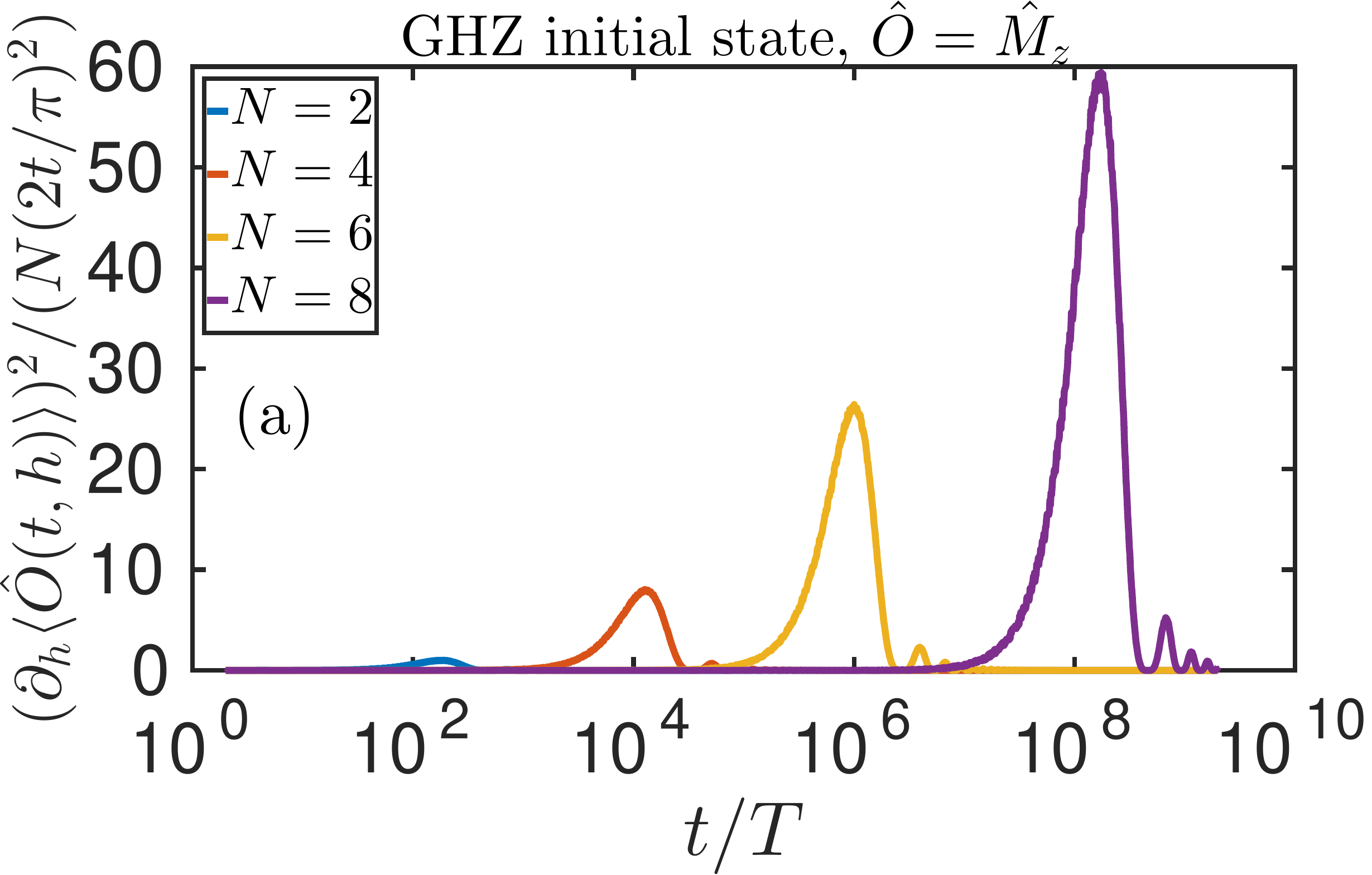}
    \includegraphics[width=0.46\linewidth]{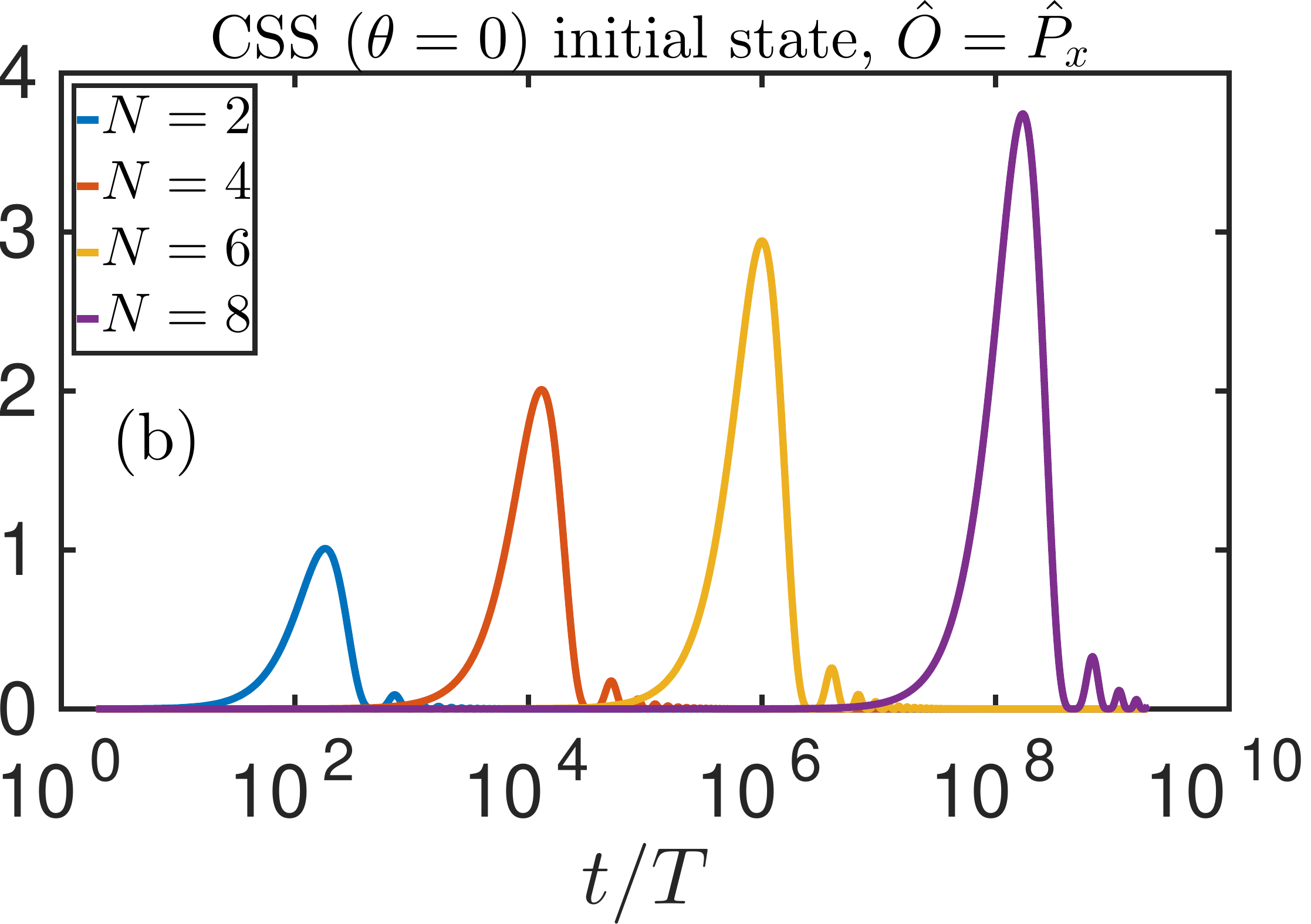}
    \includegraphics[width=0.49\linewidth]{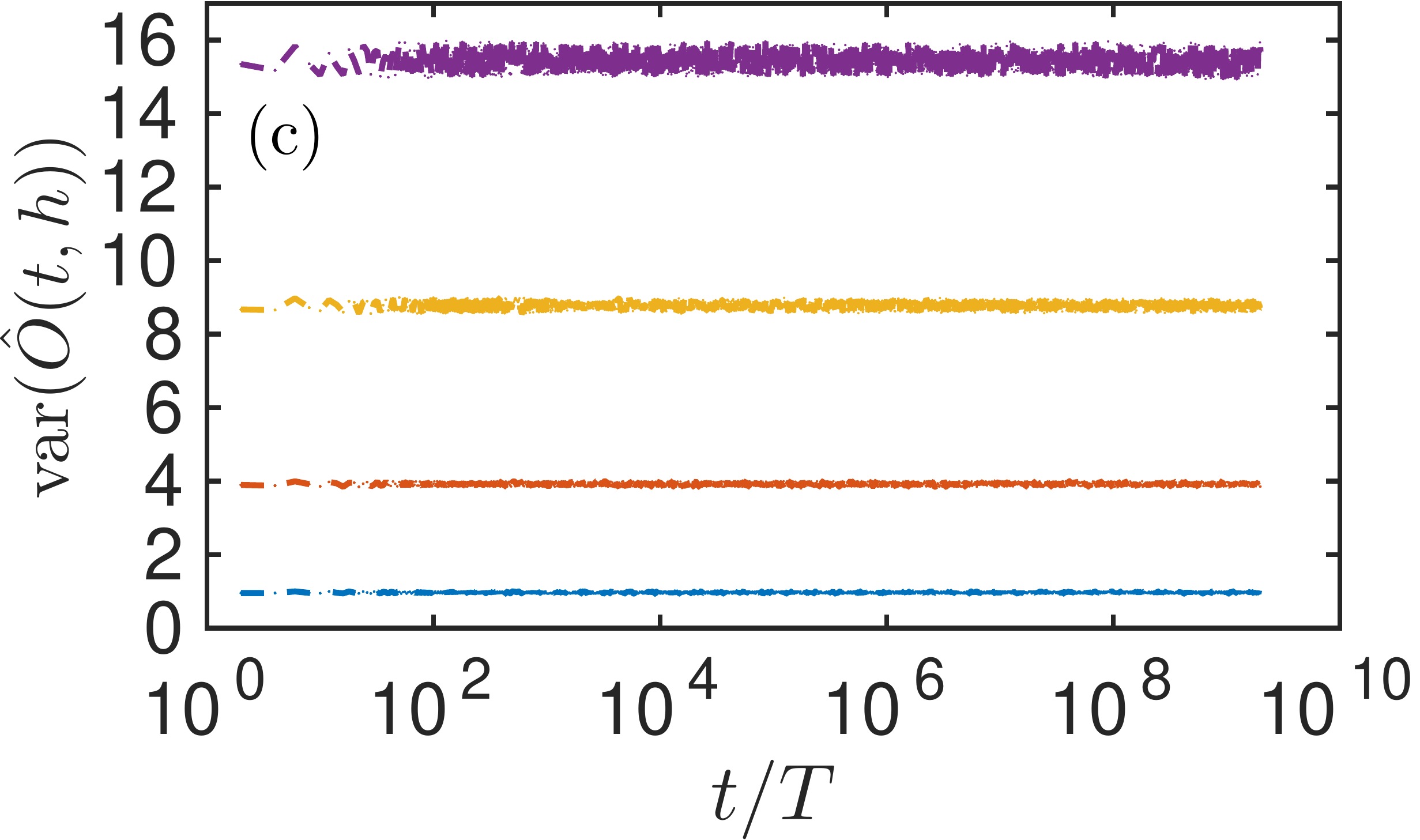}
    \includegraphics[width=0.49\linewidth]{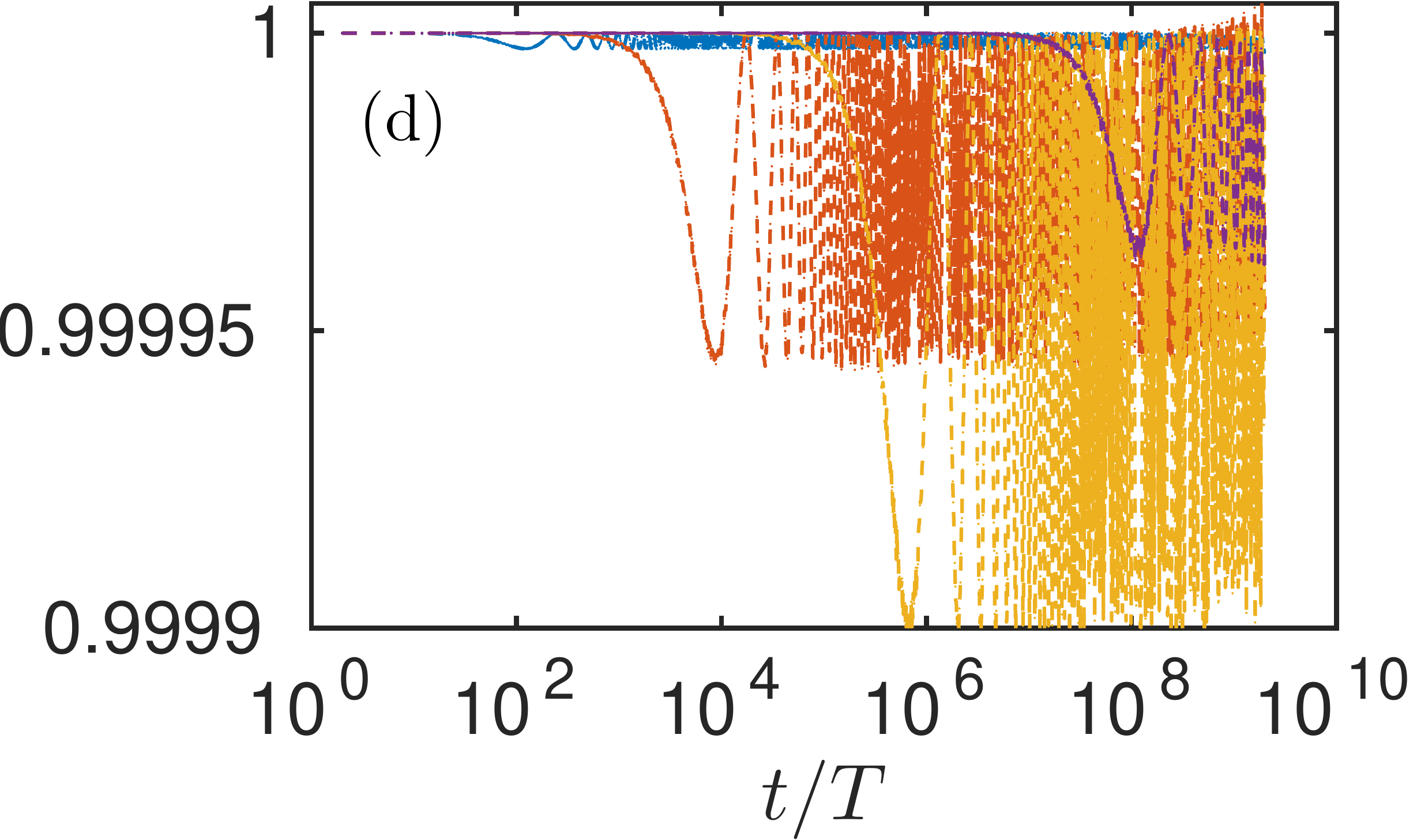}
    \caption{Decomposition of the MoM components for (left-panels) a GHZ initial state, using the magnetization ($\hat M_z$) as its observable; or (right-panel) a CSS initial state with $\theta = \phi=0$, using the parity $\hat P_x$ as its observable.  Panels (a) and (b) display the squared derivative of the expectation value $|\partial_h \langle \hat O \rangle|^2$, which captures the system's dynamic response to the external field. Panels (c) and (d) show the variance $\text{var}(\hat{O})$, which remains roughly independent of time throughout the FTC phase. In case (d), the variance is around unity, reflecting the uncorrelated nature of the spins. Meanwhile, in the GHZ case (c), the variance is significantly enhanced and scales with $N$ due to collective many-body correlations. The peaks in (a) and (b) reveal optimal measurement points in time where the accumulated phase under the periodic driving maximizes the information gain.
    }
    \label{fig:4}
\end{figure}

 Conversely, in Figs.~\ref{fig:3}(b-c), when considering the CSS initial state, the sensor initially operates within the standard quantum limit with a QFI scaling linearly with the system size $F_h \propto t^2 N$. In this regime, the sensitivity is limited by the statistical noise of uncorrelated measurements. Correlation effects emerge during the dynamics, when the Ising interactions and the kicking dynamics induce many-body correlations, enabling the system to surpass the SQL bound.
 In this context, the MoM performance remains initially suppressed compared to the GHZ case, but as the system evolves, it captures the information successfully.
 Specifically, for a CSS fully aligned along the $z$ direction ($\theta=0$), the MoM follows in close saturation to the QFI bound, with a peak concentrated at the FTC lifetime. On the other hand, for a misaligned CSS ($\theta =\pi/4$), the MoM is far from the QFI bound most of the time, reaching a peak only around the FTC lifetime.
 It  is important to notice that while the MoM reaches a peak that increases with $N$ in the aligned initial preparation ($\theta =0$), the same does not occur for the  misaligned case ($\theta =\pi/4$), exhibiting a peak that rather diminishes with increasing $N$. This observation highlights an important aspect of the MoM using the CSS, which requires a good alignment of the spins along the preferred FTC period doubling direction. Otherwise, despite having a lower cost to prepare such initial states, their effectiveness tends to remain below the Heisenberg limit for large system sizes.

\begin{figure}
    \centering
    \includegraphics[width=0.49\linewidth]{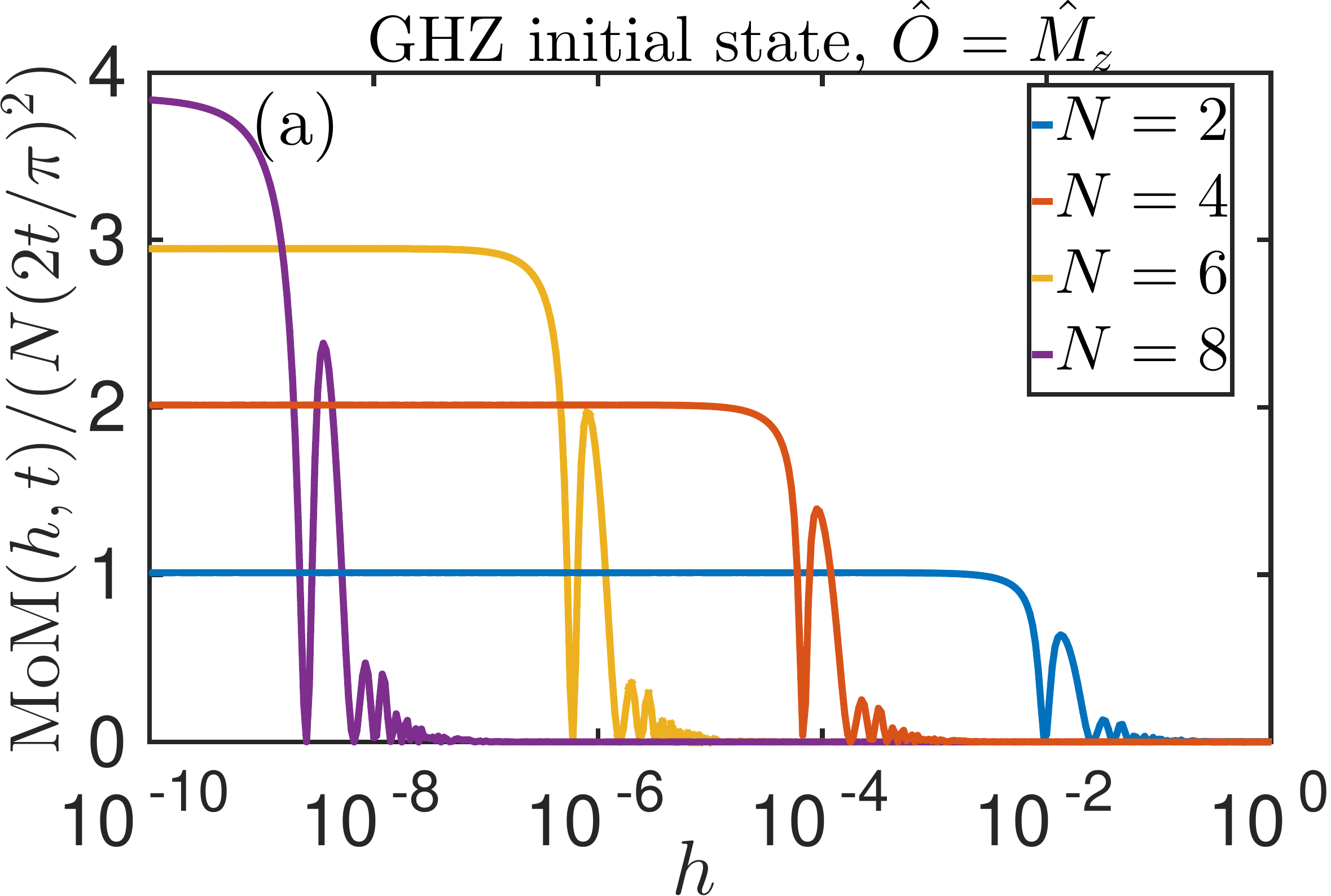}
    \includegraphics[width=0.49\linewidth]{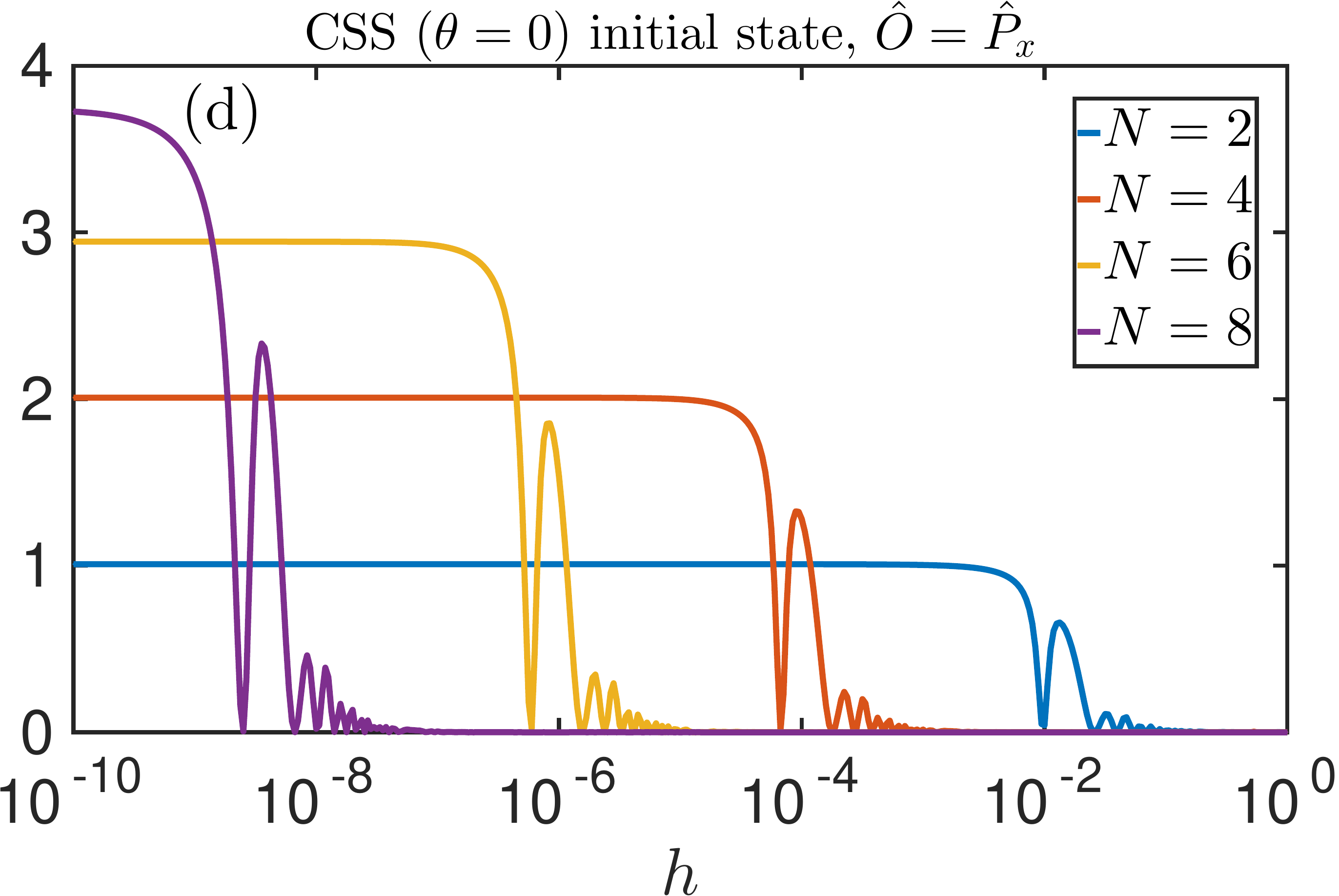}
    \includegraphics[width=0.49\linewidth]{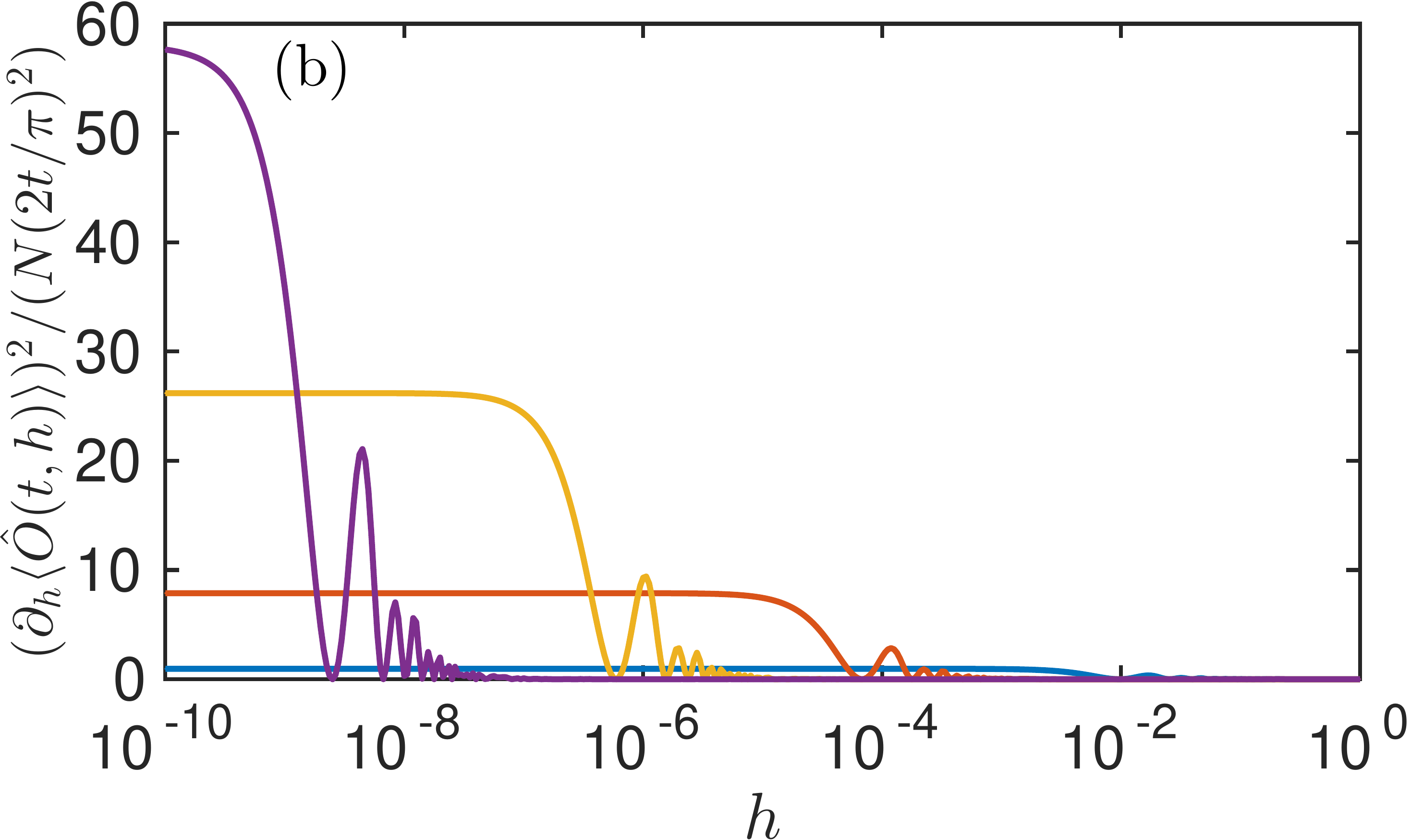}
    \includegraphics[width=0.49\linewidth]{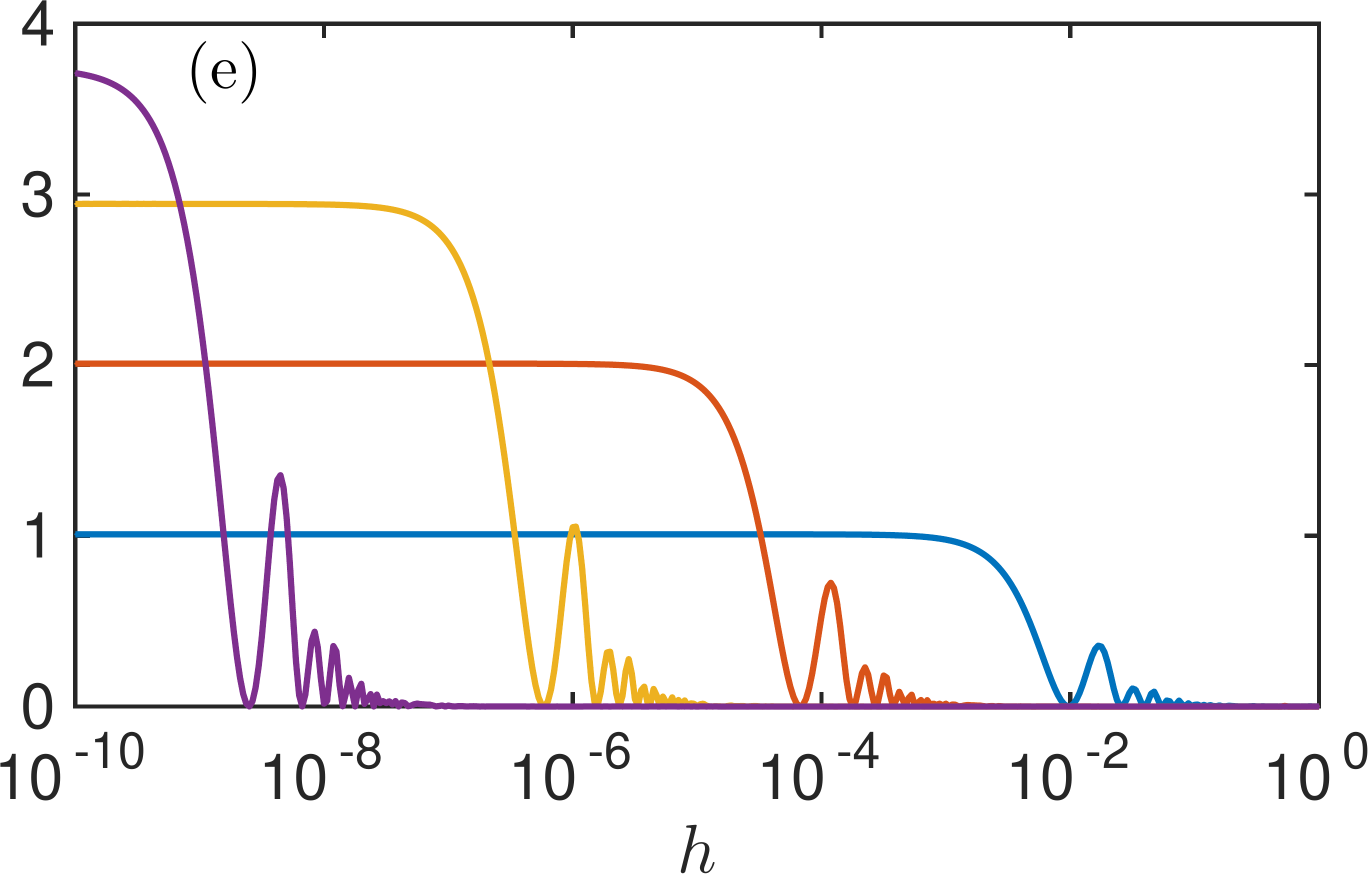}
    \includegraphics[width=0.49\linewidth]{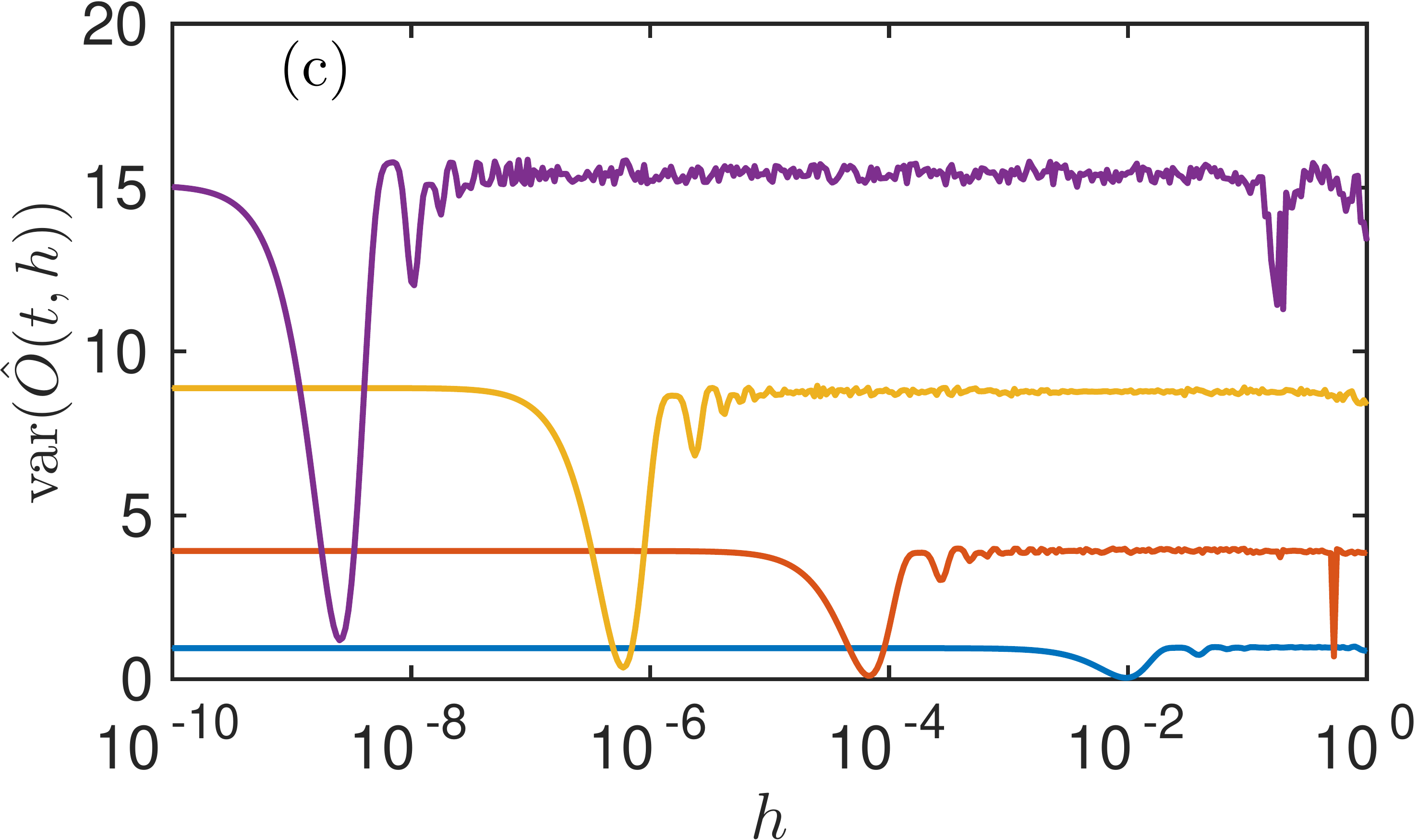}
    \includegraphics[width=0.49\linewidth]{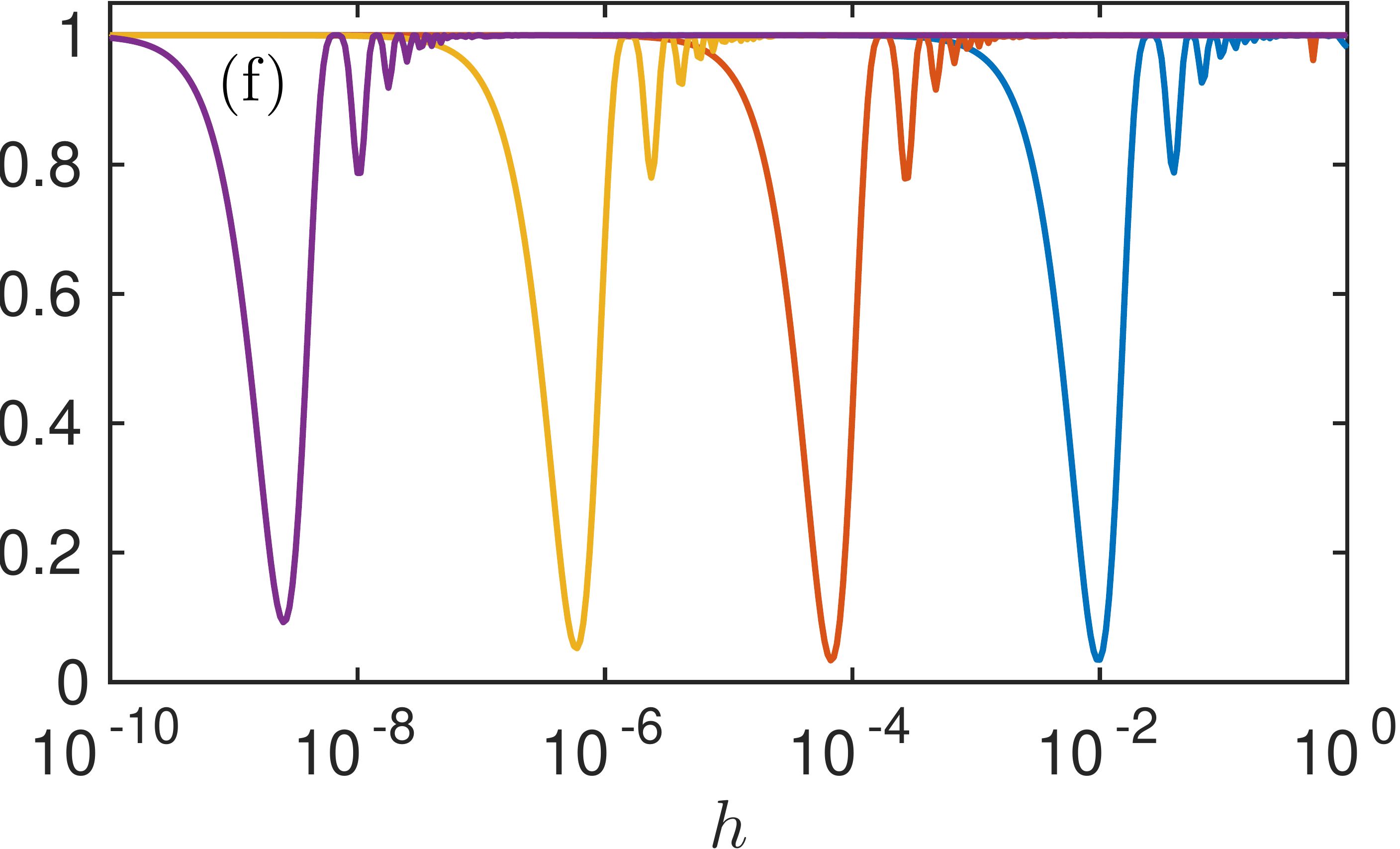}
    \caption{ MoM estimator for varying ac field amplitude $h$.
    We consider the  at a fixed time $\textrm{MoM}(h,t^*)$, with $t^*$ corresponding to its peak in the linear response regime of Fig.~\ref{fig:4}, and use its same initial states and observables.
     Panels (a) and (d) show the MoM sensitivity, exhibiting a plateau of high precision that narrows as $N$ increases. In this regime, both initial states  displays enhanced sensitivity with $N$. Panels (b) and (e) show the  derivative $|\partial_h \langle \hat O \rangle|^2$, while panels (c) and (f) illustrate the variance $\text{var}(\hat O)$.
    }
    \label{fig:5}
\end{figure}

 In summary, these results demonstrate that while the GHZ state could provide immediate, high-order sensitivity in the MoM due to its initial entanglement---as long as one could measure nontrivial observables (dressed parity-magnetization)---the CSS state relies on the system's dynamics to build correlations over time. In both cases, nevertheless, with simpler observables such as the bare parity, or magnetization, the MoM estimator could still show a Heisenberg limit scaling once it is measured on its  optimal measurement window, as given by the FTC lifetime.

\begin{figure*}
    \centering
    \includegraphics[width=0.31\linewidth]{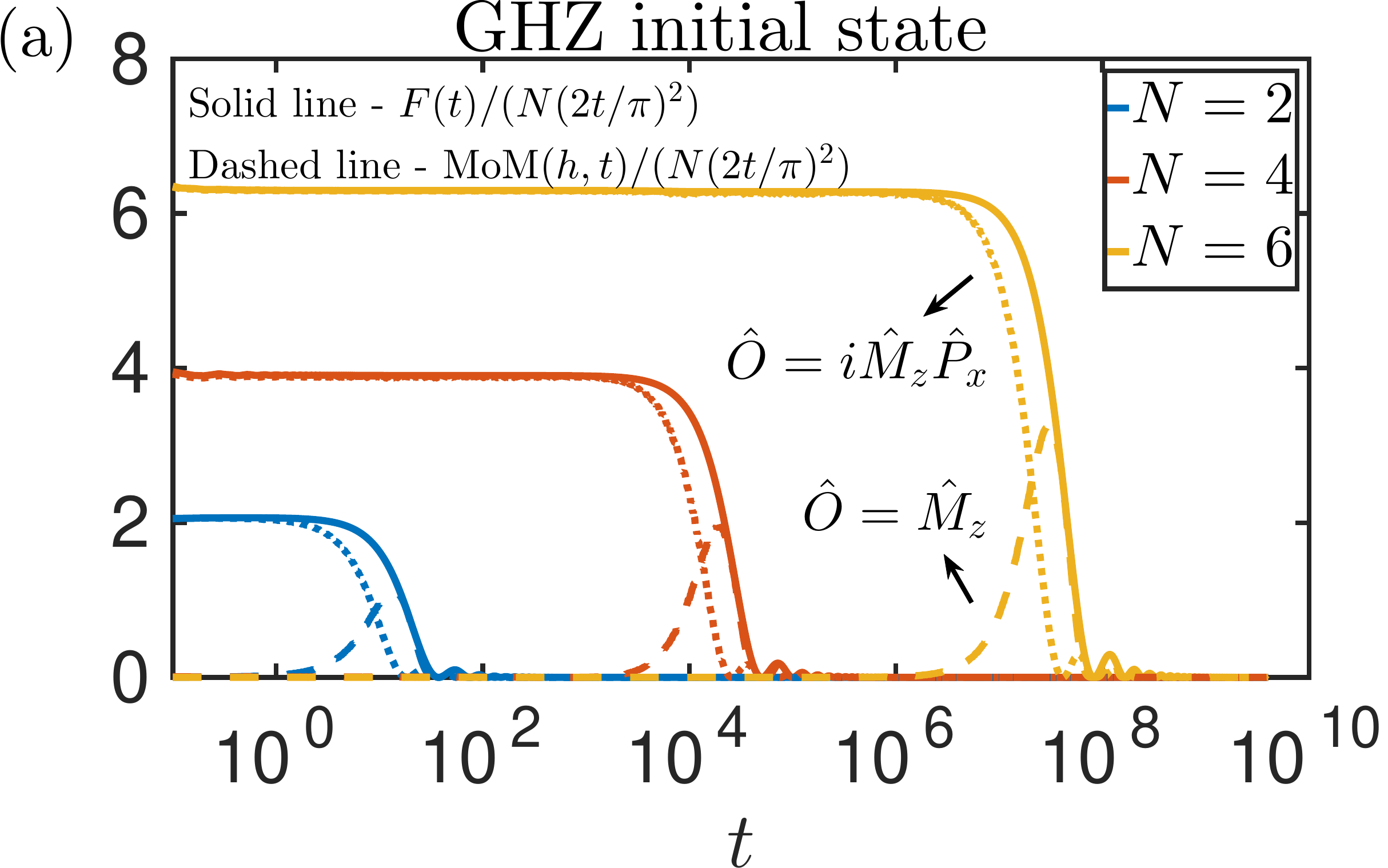}
    \includegraphics[width=0.31\linewidth]{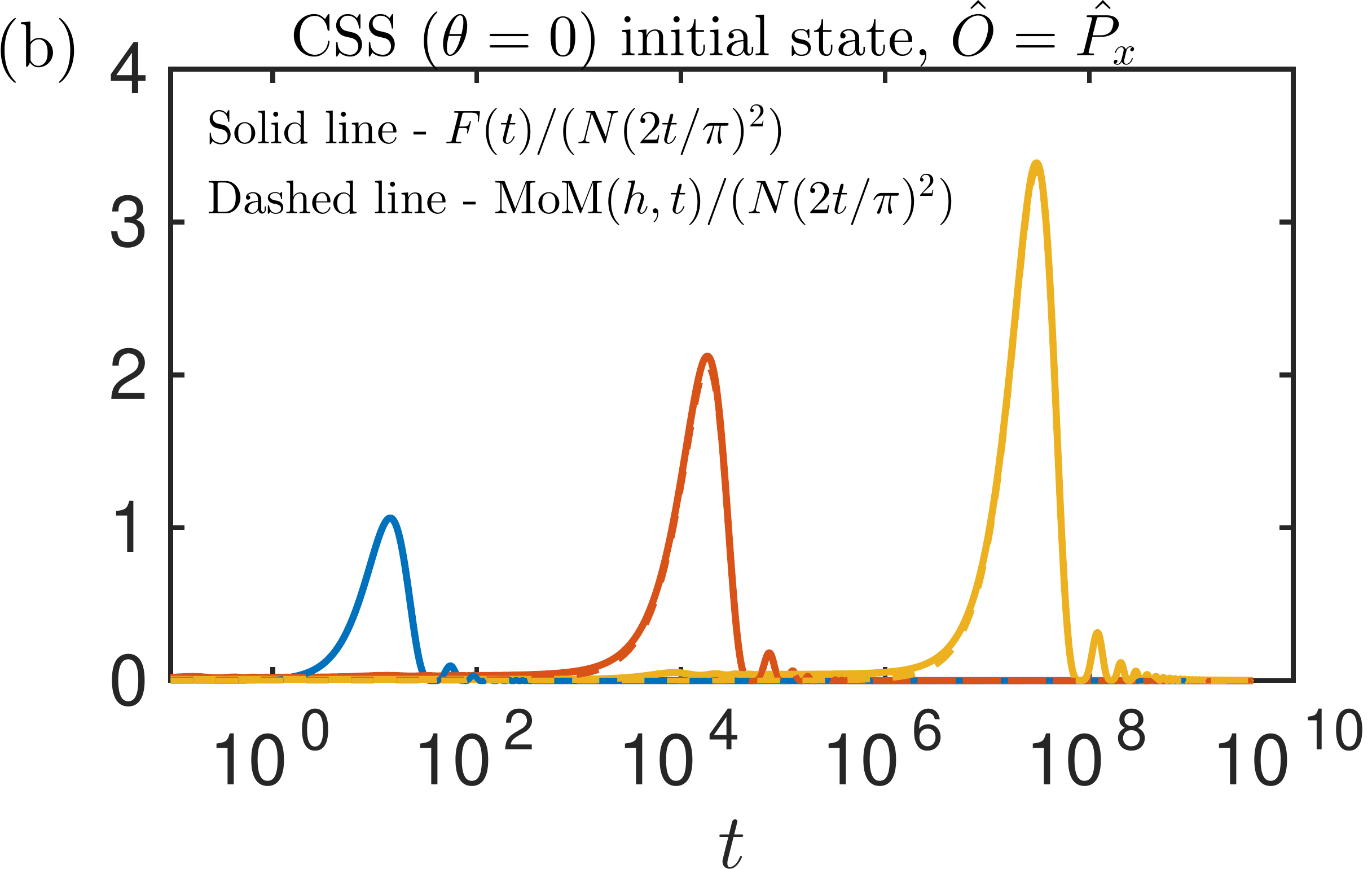}
    \includegraphics[width=0.31\linewidth]{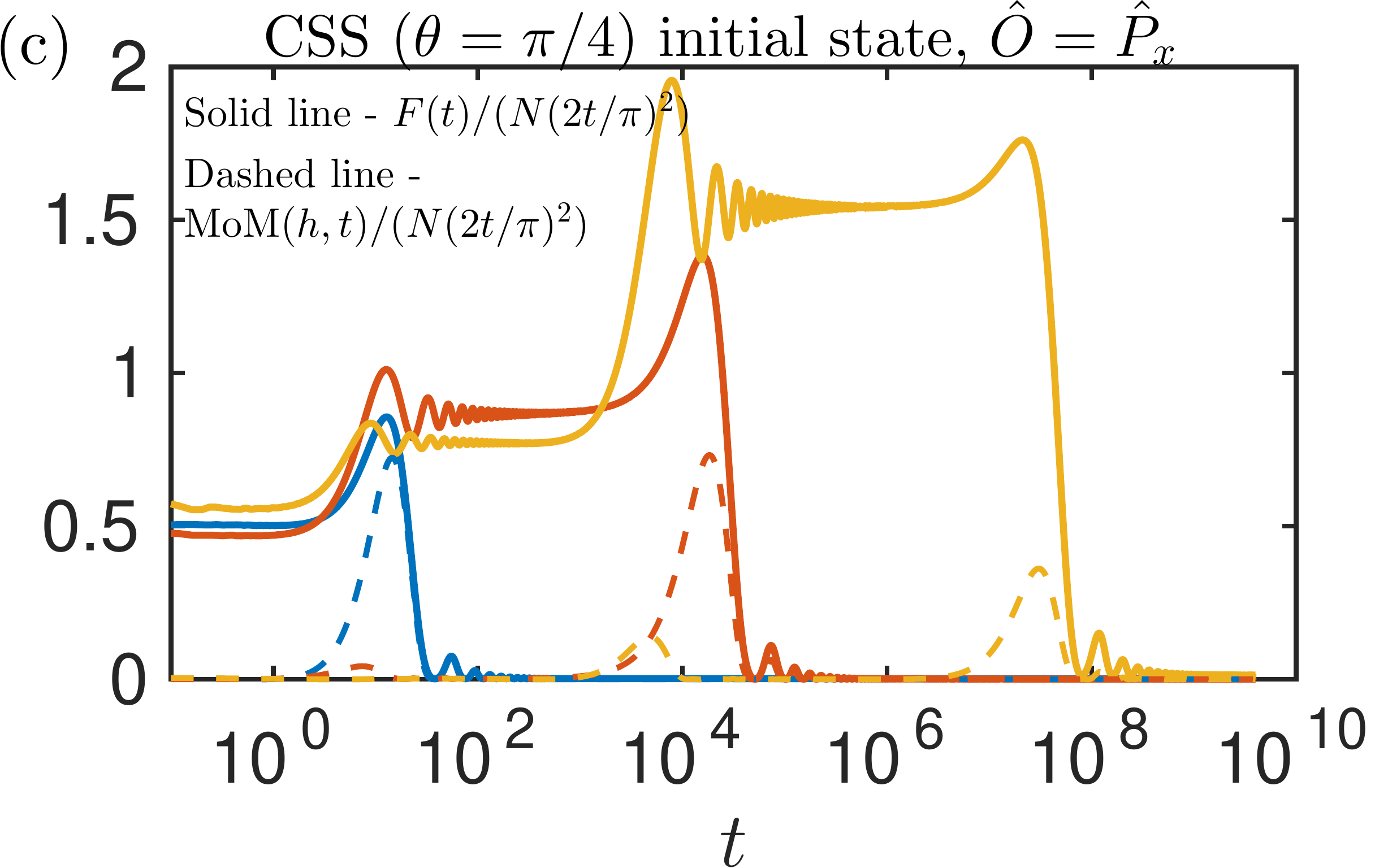}
    \caption{
    Time evolution of the  QFI (solid lines) and MoM (dashed lines)  in the linear response limit $(h \rightarrow 0)$, for  different system sizes $N$ and initial state preparations.
     The simulations use experimentally relevant parameters $JT/\hbar = 17.1$, with a period $T=0.02 \text{s}$ ~\cite{pal2018temporal}.
     We consider (a) a GHZ initial state, (b) a CSS with $\theta =\phi = 0$ and (c) a CSS with $\theta = \pi/4$, $\phi = 0$. Apart from different timescales, the behavior is qualitatively similar to those previously discussed in Fig.~\ref{fig:3}.
    }
    \label{fig:6}
\end{figure*}

\subsection*{Decomposition of MoM}

    The behavior of the MoM estimator observed in Fig.~\ref{fig:3}  is principally determined by the fluctuations of the expectation value of the chosen observables with respect to $h$.
    We show in Fig.~\ref{fig:4} how these fluctuations evolve in time. Specifically, how the dynamic response and noise components of the MoM behave dynamically. As observed in Fig.~\ref{fig:4}-(a),(b), the peaks observed in MoM are related to the high fluctuations of the response term $|\partial_h\langle O \rangle|^2$, while the variance of the observables remains roughly constant along all the dynamics---see Fig.~\ref{fig:4}-(c),(d).
    Notice that, while the GHZ variance scales with the system size $N$ due to many-body correlations, the CSS variance approaches unity, due to the uncorrelated nature of the spins.
    Nevertheless, for the GHZ state, the peak in the response derivatives terms scale by $N^{2+ \alpha} $ with $\alpha > 1$, enabling the Heisenberg limit for the MoM despite its large observable variance.

Complementing the previous analysis of MoM performance, we also examine its dependence on the field amplitude $h$. Fig.~\ref{fig:5} illustrates the evolution of the sensor sensitivity  as $h$ varies within the range $[10^{-12},10^{0}]$. Here we consider the estimator at a fixed time $\textrm{MoM}(h,t^*)$, with $t^*$ corresponding to its peak in the linear response regime of Fig.~\ref{fig:4}.
In both initial preparations, GHZ and CSS, the sensitivity forms a well-defined plateau that increases in height with the size $N$. The operational range of $h$ is constrained by the inverse of the coherence lifetime observed in the QFI within linear response. Panels (b) and (e) show the squared derivative of the observable with respect to $h$; both cases show a plateau region that is the same as panel (a) and (b). This region shows greater fluctuations, with larger changes in GHZ than in CSS. In panel (c), the variance of GHZ increases with the size $N$, and even with this increase in variance, the square of the derivative of the observable fluctuation is significantly large, so that it cancels the growth of the variance and still improves the sensor's accuracy.
These results establish a clear theoretical advantage for entangled initial states. To further validate these findings, it is crucial to examine how this protocol performs within a concrete experimental platform, as discussed in the next section.

\subsection{Experimental set-up, for Acetonitrile molecule}

To observe how the theoretical behaviors of QFI and MoM manifest at the laboratory scale, we perform a simulation using the experimental parameter values for an acetonitrile molecule~\cite{pal2018temporal} in a room-temperature environment. The  initial state preparation, in either a GHZ or the CSS states, is performed through a sequence of radio-frequency (RF) pulses followed by a free evolution over a period $\tau$ equal to the kick period.
The set of experimental parameter values are  $JT/\hbar = 17.1$, with a period $T=0.02 \text{s}$, $\hbar $ is the reduced Planck constant.
The kicking and disorder terms can be chosen within the same range as in the theoretical studies of the previous subsection.

We show our results in Fig.~\ref{fig:6}. We first observe that, under the employed experimental parameters, the qualitative behavior of both the QFI and the MoM is similar to our previous theoretical analysis across all initial states and observables considered. However, an important aspect deserves emphasis: the lifetime of the FTC is significantly increased---most notably for $N = 6$, where it becomes two orders of magnitude larger. This enhancement directly impacts the MoM when measured via bare magnetization or parity observables, as these quantities exhibit an optimal time window that roughly coincides with the FTC lifetime. Recalling the experimental constraint imposed by the molecular relaxation time, approximately $t \sim 10^{1}-10^{2}$~s,  this implies that---except for small molecules ($N \sim 2$)---the optimal time window for the MoM could not be accurately accessed experimentally. Interesting perspectives therefore lie either in shortening the FTC lifetime in such NMR experiments, for instance by introducing additional interactions into the model Hamiltonian so that the MoM becomes feasible at its optimality, or in designing methods to measure nontrivial observables---such as the dressed parity-magnetization---for which the MoM can saturate the QFI throughout the entire time evolution.

\section{Conclusions }
\label{sec.conclusions}

In this work, we demonstrate how to identify and implement near-optimal observables for quantum metrology using the MoM, with a focus on FTC sensors. First, we established a MoM protocol based on the SLD observable, and proved that it saturates the QFI bound. However, the SLD operator is typically non-local and complex, hindering its direct experimental implementation. Therefore, by exploring its structure in an FTC-based ac field sensing model, we demonstrated that it can be accurately approximated by simple observables. Specifically,  by using the simplifications associated with exponentially small gaps in cat-paired FTC subspaces and showing the relation between SLD and HSO operators,  we demonstrated that the SLD reduces to simpler observables in the MoM: a parity operator for a polarized initial state,  a parity or dressed parity-magnetization for GHZ-like states, or simply a bare magnetization for GHZ states once focusing on the FTC lifetime window---as summarized in Table~\ref{tab:summary.observables}. These results establish a practical route toward near-optimal metrology in FTC sensors, where the inaccessible SLD operator can be replaced by simpler observables while retaining quantum-enhanced sensitivity.

Indeed, we explored our theoretical predictions within a specific NMR model exhibiting a FTC phase. Our numerical simulations corroborate the proposed predictions, demonstrating how simple observables within the MoM framework---such as the bare magnetization or parity observables---can saturate the QFI bound for different initial state preparations, thereby circumventing the need for the direct implementation of SLD operator. An important aspect, however, is that in many cases the optimal performance of the MoM is constrained to a specific observation time window proportional to the FTC lifetime, which may pose a challenge for experimental realization. Indeed, using experimentally motivated parameters in our simulations, and a typical initial GHZ or CSS state, we observed that this time window can be considerably longer than the typical relaxation time of conventional NMR molecules. Consequently, strategies to overcome this limitation should be envisioned, either by shortening the FTC lifetime in such platforms, implementing nontrivial observables---such as the dressed parity-magnetization---or atypical GHZ initial states.

\begin{acknowledgments}
 We acknowledge interesting discussions with Alexandre M. Souza and Ivan S. Oliveira. M.M, A.T and F.I. acknowledge financial support from the Brazilian funding agencies CAPES, CNPq (308637/2022-4), FAPERJ (No. E-26/210.236/2024, and No.E-26/204.340/2025), and by the Serrapilheira Institute (grant number Serra 2211-42166).
\end{acknowledgments}

\bibliography{biblio}

\widetext
\clearpage
\begin{center}
	\large \textbf{Supplemental Material} \\ \vspace{0.3cm}
    M. A. Manya, Andrei Tsypilnikov, Fernando Iemini

\end{center}
\setcounter{section}{0}
\setcounter{equation}{0}
\setcounter{figure}{0}
\setcounter{table}{0}
\setcounter{page}{1}
\makeatletter
\renewcommand{\theequation}{S\arabic{equation}}
\renewcommand{\thefigure}{S\arabic{figure}}

In this Supplemental Material, we provide properties of $\hat{\mathcal S_z}$ operators, numerical analysis of matrix elements of LMG model and single spin kick-model derivation for SLD operator.

\section{Properties of signal operator in cat-subspaces}
\label{sec:sz-real}

We establish two properties of the matrix elements of the signal operator $\hat{\mathcal S}_z$ in the Floquet eigenbasis that are used in the SLD expansion of the main text: (i) the diagonal elements vanish, and (ii) the cat-pair off-diagonal element $\mathcal{O}_{i\bar{i}}\equiv\langle E_i|\hat{\mathcal S}_z|E_{\bar{i}}\rangle$ is real.

Let $\hat X$ be the parity (kick) operator and $\hat{\mathcal S}_z$ the order parameter. As in the Ref.~\cite{Andrei2026}, they satisfy
\begin{equation}\label{eq:sz-real-relations}
\hat X^2=\mathbb{I},\qquad [\hat X,\hat H_F]=0,\qquad \{\hat X,\hat{\mathcal S}_z\}=0,
\end{equation}
where $\hat H_F$ is the Floquet Hamiltonian. Since $\hat X$ commutes with $\hat H_F$, the two share an eigenbasis $\{\ket{E_i}\}$,
\begin{equation}\label{eq:sz-real-spectrum}
\hat H_F\ket{E_i}=E_i\ket{E_i},\qquad \hat X\ket{E_i}=p_i\ket{E_i},\qquad p_i=\pm1,
\end{equation}
with $p_i$ the ``parity'' of the $i$'th eigenstate; $\ket{E_{\bar{i}}}$ denotes the quasi-degenerate parity partner of $\ket{E_i}$, so that $p_{\bar{i}}=-p_i$. The anticommutator in Eq.~\eqref{eq:sz-real-relations} together with $\hat X^2=\mathbb{I}$ gives the operator identity
\begin{equation}\label{eq:sz-conj}
\hat X\hat{\mathcal S}_z\hat X=-\hat{\mathcal S}_z .
\end{equation}

\begin{prop}[Vanishing diagonal and parity selection rule]
\label{prop:diag}
The signal connects only states of opposite parity for cat states; in particular its diagonal elements vanish,
\begin{equation}\label{eq:sz-diag}
\langle E_i|\hat{\mathcal S}_z|E_j\rangle=0 \quad\text{unless}\quad p_i p_j=-1,
\qquad\text{hence}\qquad
\langle E_i|\hat{\mathcal S}_z|E_i\rangle=0 .
\end{equation}
\end{prop}
\begin{proof}
Inserting Eq.~\eqref{eq:sz-conj} between eigenstates and using $\hat X\ket{E_i}=p_i\ket{E_i}$,
\begin{equation}
\langle E_i|\hat{\mathcal S}_z|E_j\rangle
=-\langle E_i|\hat X\hat{\mathcal S}_z\hat X|E_j\rangle
=-p_i p_j\,\langle E_i|\hat{\mathcal S}_z|E_j\rangle .
\end{equation}
For $p_i p_j=+1$ this forces the element to vanish. The diagonal is the case $i=j$ with $p_i^2=1$, so $\langle E_i|\hat{\mathcal S}_z|E_i\rangle=0$.
\end{proof}

\begin{prop}[Vanish of imaginary part]
\label{prop:real}
The off-diagonal element of cat states within a quasi-degenerate doublet is real, $\textrm{Im}(\mathcal{O}_{i\bar{i}})=0$.
\end{prop}
\begin{proof}
In the symmetry-broken phase the doublet decomposes into polarized states $|\Uparrow_i\rangle$ and $|\Downarrow_i\rangle\equiv\hat X|\Uparrow_i\rangle$, with $\langle\Uparrow_i|\Downarrow_i\rangle=0$,
\begin{equation}\label{eq:sz-real-cat}
\ket{E_i}=\frac{|\Uparrow_i\rangle+|\Downarrow_i\rangle}{\sqrt2},\qquad
\ket{E_{\bar{i}}}=\frac{|\Uparrow_i\rangle-|\Downarrow_i\rangle}{\sqrt2},
\end{equation}
where $|\Uparrow_i\rangle$ carries a definite macroscopic magnetization $\hat{\mathcal S}_z|\Uparrow_i\rangle=s_i|\Uparrow_i\rangle$, $s_i\in\mathbb{R}$. The anticommutation fixes $\hat{\mathcal S}_z|\Downarrow_i\rangle=\hat{\mathcal S}_z\hat X|\Uparrow_i\rangle=-\hat X\hat{\mathcal S}_z|\Uparrow_i\rangle=-s_i|\Downarrow_i\rangle$, so $|\Uparrow_i\rangle$ and $|\Downarrow_i\rangle$ are eigenstates of $\hat{\mathcal S}_z$ with opposite eigenvalues $\pm s_i$. Acting on the partner state,
\begin{equation}\label{eq:sz-real-swap}
\hat{\mathcal S}_z\ket{E_{\bar{i}}}=\frac{1}{\sqrt2}\bigl(s_i|\Uparrow_i\rangle- (-s_i)|\Downarrow_i\rangle\bigr)=s_i\ket{E_i},
\end{equation}
and therefore
\begin{equation}\label{eq:sz-real-main}
\mathcal{O}_{i\bar{i}}=\langle E_i|\hat{\mathcal S}_z|E_{\bar{i}}\rangle=s_i\in\mathbb{R},
\qquad
\textrm{Im}(\mathcal{O}_{i\bar{i}})=0.
\end{equation}
\end{proof}

The cat-pair element is thus real and equal to the macroscopic magnetization $s_i$, and Eq.~\eqref{eq:sz-real-swap} shows that $\hat{\mathcal S}_z$ flips the parity of the cat state while preserving the doublet, the hallmark of the time-crystal response. The identity is exact for fully polarized references such as $|\Uparrow_i\rangle=\ket{\uparrow\uparrow\cdots\uparrow}$, and for finite-size doublets it holds up to corrections that vanish in the thermodynamic limit.

\section{Analysis of observables in Hamiltonian basis for LMG model}

From the semiclassical approach for LMG model~\cite{Andrei2026} we know operators can be presented as
\begin{equation}\label{eq:sm-average-magnetization}
	\begin{aligned}
		\hat{\mathcal S}_z = \sqrt{1 - \frac{B^2}{J^2}}\, \hat{\mathcal S}_z' - \frac{B}{J}\, \hat{\mathcal S}_x'.
	\end{aligned}
\end{equation}
where
\begin{equation}
	\begin{aligned}
		\hat{\mathcal S}_x' & = \frac{\sqrt{N}}{2} \left(1 + 2 (B/J)^2 \right)^{-1/4} \left( \hat b + \hat b^\dagger \right) + O\left(N^{-3/2}\right),
		\\
		\hat{\mathcal S}_z' & = \frac{N}{2} + \frac{1}{2} \left(1 - \frac{1 + (B/J)^2}{\sqrt{1 + 2 (B/J)^2}}  \right) - \frac{1 + (B/J)^2}{\sqrt{1 + 2 (B/J)^2}} b^\dagger b + \frac{1}{2} \frac{(B/J)^2}{\sqrt{1 + 2 (B/J)^2}} \left(b^\dagger b^\dagger + b b \right)
	\end{aligned}
\end{equation}

so, we can see the following non-diagonal elements in eigenvalues basis,
\begin{equation}
	\begin{aligned}
        &\langle E_n | \hat{\mathcal S}_z |E_{\bar{n}} \rangle = \langle \Uparrow_n | \hat{\mathcal S}_z |\Uparrow_n \rangle = O(N), \\
		&\langle E_n | \hat{\mathcal S}_z |E_{\bar{n}\pm1} \rangle = \frac{1}{2}\left(1 + p_n p_{n\pm1} \right) \langle \Uparrow_n | \hat{\mathcal S}_z |\Uparrow_{n\pm1} \rangle  = O\left(\sqrt{N}\right), \\
        &\langle E_n | \hat{\mathcal S}_z |E_{\bar{k}} \rangle = O(1), \quad |n - k| \ge 2.
    \end{aligned}
\end{equation}

It can be confirmed numerically by computing the matrix elements $|\bra{E_i}\hat O\ket{E_j}|$ for $\hat O \in \{\hat X, \hat{\mathcal S}_x, \hat{\mathcal S}_y, \hat{\mathcal S}_z\}$ and inspecting the heatmaps shown in Fig.~\ref{fig:obs-eigenbasis} for  $N=40$, at fixed $B/J=0.4$.

\begin{figure*}[h]
    \centering
    \includegraphics[width=\linewidth]{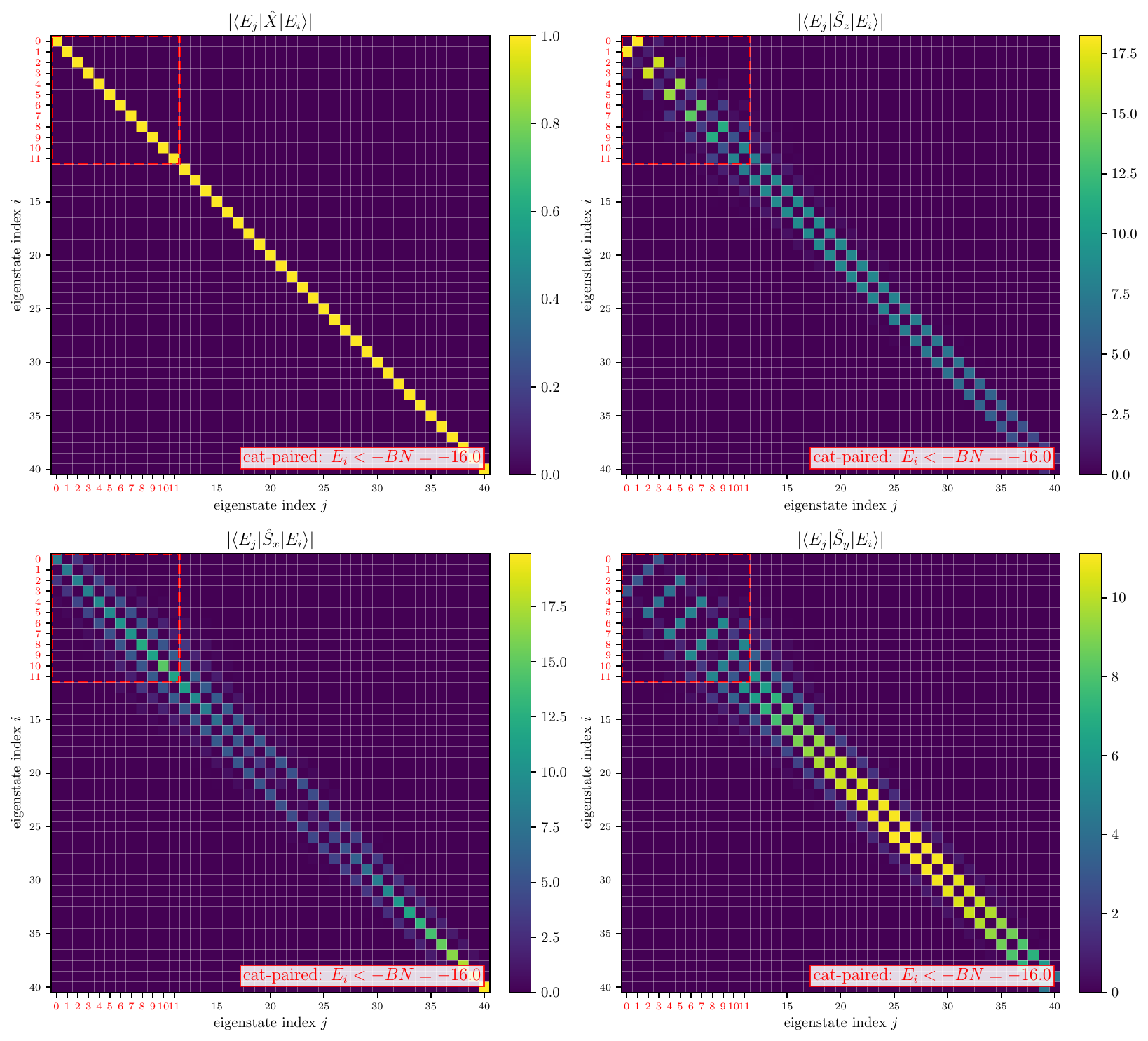}
    \caption{Modulus of the matrix elements $|\langle E_j|\hat O|E_i\rangle|$ in the eigenbasis of pure LMG model~\cite{Andrei2026} $\hat H_0 = -(2J/N)\hat{\mathcal S}_z^2 - 2B\hat{\mathcal S}_x$, for $\hat O \in \{\hat X, \hat{\mathcal S}_z, \hat{\mathcal S}_x, \hat{\mathcal S}_y\}$, where $\hat X = e^{-i\pi \hat{\mathcal S}_x}$ is the parity operator. The dashed red box encloses the cat-paired subspace of states with $E_i < -BN$, inside which nearly degenerate pairs of opposite parity are coupled by the kick generator $\hat{\mathcal S}_x$. Indices $(j,i)$ are zero-based and ordered by ascending $\hat H_0$ eigenvalue. Calculations done for LMG model with $N=40$, at fixed $B/J=0.4$}
    \label{fig:obs-eigenbasis}
\end{figure*}

Symmetries fix most of the diagonal: $[\hat H_F,\hat X]=0$ implies $\langle E_i|\hat X|E_i\rangle=\pm 1$, and since $\{\hat{\mathcal S}_z,\hat X\}=\{\hat{\mathcal S}_y,\hat X\}=0$ the diagonals of $\hat{\mathcal S}_z$ and $\hat{\mathcal S}_y$ vanish in any nondegenerate parity sector. Only $\hat{\mathcal S}_x$ has an allowed nonzero diagonal, which grows as $O(N)$.

\section{SLD and single spin model}

    SLD operator for pure state defined as
    \begin{equation}
        \label{eq:sld-definition-pure-states-SM}
        \hat{L}_h = 2\Big(\ket{\partial_h \psi_h} \bra{\psi_h} + \ket{\psi_h} \bra{\partial_h \psi_h}\Big)
    \end{equation}
    can be computed analytically for single spin toy model with Hamiltonian
    \begin{eqnarray}
        \hat{H}_{\text{ss}}(h) =  -\frac{\Delta}{2} \hat{\sigma}^x + h \theta(t) \hat{\sigma}^z - \frac{\phi}{2} \sum_{n=1}^{\infty} \delta(t - nT) \hat{\sigma}^x,
    \end{eqnarray}
    with $\Delta>0$ the bare gap and $\phi=\pi$ the kick angle. The ac signal is given by a step function in period-doubling resonance to the kicking, i.e. $\theta(t+kT) = (-1)^k \theta(t)$ with $|\theta(t)|=1$ for all $t$.

    The eigenvectors of $\hat{H}_{\text{ss}}|_{t=0,h=0}=-(\Delta/2)\hat\sigma^x$, written in the computational basis $\hat{\sigma}^z$, are the $\hat\sigma^x$ eigenstates
    \begin{equation}
        \begin{aligned}
        \ket{E_+} &= \frac{1}{\sqrt{2}}
        \begin{pmatrix}
            1 \\ 1
        \end{pmatrix},  \quad &E_{+}=-\frac{\Delta}{2},  \\
        \ket{E_-}  &= \frac{1}{\sqrt{2}}
        \begin{pmatrix}
            1 \\ -1
        \end{pmatrix}, \quad &E_{-}=\frac{\Delta}{2},
       \end{aligned}
    \end{equation}
    so that $\hat\sigma^x\ket{E_\pm}=\pm\ket{E_\pm}$ and $\hat\sigma^z\ket{E_\pm}=-\ket{E_\mp}$, hence $\ket{E_+}$ is the lower-energy eigenstate of the bare Hamiltonian.

    \subsection{Evolution of single spin model}

    Unitary evolution $\hat {U} (t_1;t_2) \equiv e^{-i \hat H_{\text{ss}}(t_1-t_2)}$ using the Floquet theorem, during the $n$-th Floquet period is given by
    \begin{equation}
        \hat {U} \left((n+1)T ; nT\right) = \hat {U}_{F}(h) = \hat{X}\, e^{-i \hat{H}_0((-1)^n h)\, T },
        \qquad
        \hat{H}_0(h) = -\tfrac{\Delta}{2}\hat{\sigma}^x + h\hat{\sigma}^z,
    \end{equation}
    where the ``kick operator'' is $\hat X = e^{i \frac{\phi}{2} \hat{\sigma}^x} =  i \hat{\sigma}^x$ for $\phi=\pi$, and $\Delta$ is the bare gap.

    \subsubsection{Single spin case SLD operator representation in Pauli basis}
    \textbf{Proof} that for $n \geq 1$,
    \[
         \hat {U} \left(nT ; 0\right) = \hat{X} \hat{U}_{n-1}\, \hat{X} \hat{U}_{n-2}\, \cdots\, \hat{X} \hat{U}_1\, \hat{X} \hat{U}_0 \;=\; \hat{X}^{n}\, \hat{U}_+^{\,n},
        \qquad \hat{U}_k \equiv \hat{U}_F((-1)^k h).
    \]
    \begin{proof}
        Using the following property $U_k\, X = X\, U_{k+1}$. \quad $(\star)$

        \textbf{Base ($n=1$).} $X U_0 = X U_+ = X^1\, U_+^{\,1}$. \checkmark

        \textbf{Inductive step ($n-1 \to n$).} Assume
        \[
            X U_{n-2}\, X U_{n-3}\, \cdots\, X U_0 \;=\; X^{\,n-1}\, U_+^{\,n-1}. \tag{IH}
        \]
        Split off the leftmost block:
        \[
            X U_{n-1}\, X U_{n-2}\, \cdots\, X U_0
            \;\stackrel{\text{(IH)}}{=}\; X U_{n-1}\, \cdot\, X^{\,n-1}\, U_+^{\,n-1}.
        \]
        Apply $(\star)$ a total of $n-1$ times to slide $U_{n-1}$ through the block $X^{\,n-1}$:
        \[
            U_{n-1}\, X^{\,n-1} \;=\; X^{\,n-1}\, U_{2(n-1)} \;=\; X^{\,n-1}\, U_+,
        \]
        since $2(n-1)$ is even. Therefore
        \[
            X U_{n-1}\, X^{\,n-1}\, U_+^{\,n-1}
            \;=\; X \cdot X^{\,n-1}\, U_+ \cdot U_+^{\,n-1}
            \;=\; X^{n}\, U_+^{\,n}.
        \]
    \end{proof}

    Collecting the kicks through the property above, the full Floquet evolution after $n$ steps is
    \begin{equation}\label{eq:Un}
        \hat{U}_{h\to0}(nT) = \hat{X}^n\, e^{-i n T \hat{H}_0},\qquad \hat{H}_0\equiv\hat{H}_0(h).
    \end{equation}
    The single-period evolution matrix reads
    \begin{eqnarray*}
        \begin{aligned}
            e^{-i t \hat{H}_0} = \cos\!\left(\tfrac{t}{2}\Omega_h\right)\hat{\mathbb{I}} - 2i\frac{h \sin\!\left(\tfrac{t}{2}\Omega_h\right)}{\Omega_h}\hat{\sigma}^z + i\frac{ \Delta \sin\!\left(\tfrac{t}{2}\Omega_h\right)}{\Omega_h}\hat{\sigma}^x,
        \end{aligned}
    \end{eqnarray*}
    The stroboscopic evolution then has the form
    \begin{eqnarray}
        \begin{aligned}
            \hat{U}_{h\to0}(nT) &= \hat{X}^n \left(e^{-i T \hat{H}_0}\right)^n = \hat{X}^n e^{-i n T \hat{H}_0} \\
            &= (i\hat{\sigma}^x)^n \left( \cos\!\left(\tfrac{\Delta}{2} n T\right)\hat{\mathbb{I}} + i \sin\!\left(\tfrac{\Delta}{2} n T\right)\hat{\sigma}^x - i \frac{2h \sin\!\left(\tfrac{\Delta}{2} n T\right)}{\Delta}\hat{\sigma}^z \right) +O(h^2),
        \end{aligned}
    \end{eqnarray}
    where $\Omega_h = \sqrt{\Delta^2 + 4h^2}$ and in the last line we used $\hat{X}^n = i^n (\hat{\sigma}^x)^n$ and expanded to linear order in $h$. From this form we read off, in the computational basis, the zeroth-order evolution,
    \begin{eqnarray*}
        \begin{aligned}
            \ket{\psi_h(t)}\big|_{h=0}&=\hat{U}_0(nT)\ket{\psi_0} = i^n (\hat{\sigma}^x)^n \left( \cos\!\left(\tfrac{\Delta}{2} nT\right)\hat{\mathbb{I}} +i \sin\!\left(\tfrac{\Delta}{2} nT\right)\hat{\sigma}^x \right)\ket{\psi_0}, \\
            \ket{\partial_h\psi_h(t)}\big|_{h=0}&=\partial_h\left(\hat{U}_0(nT)\right)\ket{\psi_0}\Big|_{h=0} = -2i^{n+1}\frac{\sin\!\left(\tfrac{\Delta}{2} n T\right)}{\Delta} (\hat{\sigma}^x)^n \hat{\sigma}^z \ket{\psi_0}.
        \end{aligned}
    \end{eqnarray*}
    Using Eq.~\eqref{eq:sld-definition-pure-states-SM}, one finds for the initial state
    $\ket{\psi_0(0)} = \cos{\tfrac{\theta}{2}} \ket{E_+} + e^{i \varphi}\sin{\tfrac{\theta}{2}}  \ket{E_-}$,

    \begin{equation}
        \begin{aligned}
            \hat L_{h\to0}(nT) =  \frac{2\sin\!\left(\Delta nT / 2\right)}{\Delta / 2}
            \Bigg(&-\sin\!\left( \Delta nT / 2 - \varphi\right)  \sin\theta \,\hat{\sigma}^x \\
            &+ (-1)^n \cos(\Delta nT /2)\cos\theta \,\hat{\sigma}^y   \\
            &- (-1)^n \sin(\Delta nT /2)\cos\theta \,\hat{\sigma}^z \Bigg).
        \end{aligned}
    \end{equation}
    Rewriting this operator in $\ket{E_-},\ket{E_+}$ basis we can get the same expressions for SLD in single cat-subspace as in main text with energy gap $E_+-E_-=-\Delta$.

    Using expansion
    \begin{equation}
        \hat{L}_{h\to0}(nT) = \vec{L}(nT)\cdot\vec{\sigma},
        \quad
        \vec{\sigma} = (\hat{\sigma}^x, \hat{\sigma}^y, \hat{\sigma}^z),
    \end{equation}
    with SLD projected onto the Pauli basis through $L_\alpha = \tfrac{1}{2}\Tr[\hat{L}\hat{\sigma}^\alpha]$ we can get the expression for dynamics of SLD on the Bloch sphere.

    These results show that for the polarized state ($\theta=\pi/2$, $\varphi=0$), the optimal observable is aligned with $i \hat{\sigma}^x$ (parity operator), while for the cat state ($\theta=0$) the dynamics exhibits a period-doubled rotation in the $y$--$z$ plane.

\end{document}